\documentclass[12pt,draftclsnofoot,onecolumn]{IEEEtran}
\usepackage[cmex10]{amsmath}
\usepackage{amssymb}
\usepackage{cite}
\usepackage{amsmath}
\usepackage{wasysym}
\usepackage{booktabs}
\usepackage{paralist}
\usepackage{balance}
\usepackage{subfigure}
\usepackage{multicol}
\usepackage{graphicx}
\usepackage{epsfig}
\usepackage{epstopdf}
\usepackage[ruled,linesnumbered]{algorithm2e}
\usepackage{bm}
\usepackage{caption}
\usepackage{amsmath, amssymb,amsthm}
\usepackage{upgreek}
\usepackage{CJK}
\usepackage{color}
\usepackage{float}
\usepackage{xfrac}
\usepackage[T1]{fontenc}
\usepackage{amsfonts}
\usepackage{amssymb}
\usepackage{amsthm,amsmath}
\usepackage{mathrsfs}
\usepackage{indentfirst}
\usepackage{multirow}
\usepackage[ruled]{algorithm2e}
\usepackage{algorithmic}
\usepackage{lineno}
\usepackage{array}

\newtheorem{theorem}{Theorem}
\newtheorem{proposition}{Proposition}
\newtheorem{lemma}{Lemma}
\newtheorem{corollary}{Corollary}

\newtheorem{remark}{Remark}

\ifCLASSINFOpdf
\graphicspath{{pic/}}
\else
\fi
\begin{document}
%

\title{Analysis and Optimization of  Cache-Enabled mmWave HetNets with \textcolor{blue}{Integrated Access and Backhaul} }
%

\author{Chenwu~Zhang,
        Hancheng~Lu,~\IEEEmembership{Senior Member,~IEEE}, and Zhuojia Gu 

\thanks{
This work was supported in part by National Key R\&D Program of China under Grant 2020YFA0711400 and National Science Foundation of China under Grant 61771445, 61631017, 91538203. \textcolor{blue}{Part of this work has been presented at IEEE Wireless Communications and Networking Conference(WCNC), Seoul, South Korea, May, 2020\cite{WCNC20}.} 

	Chenwu~Zhang, Hancheng~Lu and Zhuojia Gu  are  with CAS Key Laboratory of Wireless-Optical Communications, University of Science and Technology of China, Hefei 230027, China. (Email: cwzhang@mail.ustc.edu.cn;  hclu@ustc.edu.cn; guzj@mail.ustc.edu.cn).

}
}

\maketitle
\date{}
\begin{abstract}

In milimeter wave heterogenous networks with \textcolor{blue}{integrated access and backhaul} (mABHetNets), a considerable part of spectrum resources are occupied by the backhaul link, which limits the performance of the access link.  In order to overcome such backhaul ``\emph{spectrum occupancy}'', we introduce cache in mABHetNets. Caching popular files at small base stations (SBSs) can offload the backhaul traffic and transfer spectrum from the backhaul link to the access link. 
To achieve the optimal performance of the cache-enabled mABHetNets, we first analyze the signal-to-interference-plus-noise ratio (SINR) distribution and derive the  average potential throughput (APT) expression by stochastic geometric tools.  Then, based on our analytical work, we formulate a  joint optimization problem of cache decision and spectrum partition to maximize the APT.  Inspired by the block coordinate descent (BCD) method, we propose a joint cache decision, spectrum partition and power allocation (JCSPA) algorithm to find the optimal solution.
Simulation results show the convergence and enhancement of the proposed algorithm. Besides, we verify the APT under different parameters and find that the introduction of cache facilitates the transfer of backhaul spectrum to access link.  
\textcolor{blue}{Jointly deploying appropriate caching capacity at SBSs and performing specified spectrum partition can bring up about 90\% APT gain in mABHetNets. }    

\end{abstract}

\begin{IEEEkeywords}
Millimeter wave, integrated access and backhaul \textcolor{blue}{(IAB)}, cache,  spectrum transfer, average potential throughput (APT). 
\end{IEEEkeywords}

%
\IEEEpeerreviewmaketitle

\section{Introduction}


The number of mobile terminal equipments  increases sharply in recent years and the demand for data  also increases at an annual growth rate of 63\% \cite{whitepaper}, while the existing scarce LTE spectrum resources\cite{IAB0} severely affect the high-demand communication. 
Fortunately,  owing  to the characteristics of abundant spectrum in milimeter wave (mmWave), the   high throughput communication has become possible by using mmWave in both the access and the backhaul link. However, mmWave signals are only suitable for short distance propagation \textcolor{blue}{and are  blockage sensitive}, so the base stations (BSs) need to be deployed densely in mmWave heterogeneous networks (HetNets), which in turn increases the backhaul traff\mbox{}ic. Particularly, integrated access and backhaul (IAB) in mmWave HetNets (mABHetNets)  has been standardized for the f\mbox{}ifth generation (5G) mobile communication technology in the third generation partnership project (3GPP) Rel-16\cite{IAB0,IAB1,IAB2,IAB3,3GPPIAB}. 
In  mABHetNets, macro base stations (MBSs) are connected to the mobile core network via high-speed fiber backhaul, but it is not practical to connect all small base sations (SBSs) to the mobile core network. High-power MBSs are overlaid by denser  low-power mmWave SBSs where SBSs provide high  rate service to the users by wireless access link and the MBS  maintains the backhaul capacity of the SBSs by the wireless backhaul link. 

 Since the access and the backhaul links share the mmWave spectrum resources in mABHetNets, spectrum partition has a significant impact on the network performance and hence attracted many research attempts\cite{Partition0,Partition2,Partition3,Partition4}. Authors in \cite{Partition0}  proposed the resource partition strategies as the total spectrum is dynamically between access and backhaul, or a static partition is defined for the access and backhaul link to achieve the maximum  coverage probability in mmWave HetNets. 
 Authors in\cite{Partition2}  explored the optimal partition of access and backhaul spectrum to maximize the rate coverage and authors in \cite{Partition3} leveraged allocated resource ratio between radio access and backhaul to study the maximization of network capacity by considering the fairness among SBSs. In\cite{Partition4}, the beamforming and spectrum partition were jointly investigated to improve the network capacity of mABHetNets. \textcolor{blue}{Maria \emph{et al.} in \cite{Partition5} optimized phase shift of Reconfigurable Intelligent Surface (RIS) element, bandwidth splitting among wireless access and backhaul, and transmission power to maximize the energy efficiency of RIS-aided IAB network.}


However, even with optimal spectrum partition, a considerable part of mmWave spectrum resources are still occupied by the backhaul link to maintain the backhaul capacity.  According to the findings in \cite{Partition2}, up to 50\% mmWave spectrum might be used
in backhaul link to satisfy the high speed data traffic. 
\textcolor{blue}{Such a \emph{``spectrum occupancy''} phenomenon of wireless backhaul link severely limits the transmission of wireless access link, which further affects the communication quality of users and degrades the spectrum resource utilization. Besides, several reports  show that a few files with high popularity are often requested by users, and this increases the transmission pressure on the wireless backhaul link \cite{2u,caching}. Repeated transmission also causes a lot of waste of power consumption and spectrum resources. Fortunately,}  enabling caching at the wireless edge such as SBSs is considered as a promising way to improve the energy efficiency and network throughput recently\cite{2u,caching,caching1,MostPop,caching2}, \textcolor{blue}{especially for those files with high popularity and frequently requested.  
	\textcolor{blue}{
		In \cite{caching}, Tao \emph{et al.} proposed a cache-enabled radio access network (RAN) to minimize the total network cost.
		Authors in \cite{caching1}  applied coded caching into small-cell network (SCN) and invetstigated average fractional offloaded traffic (AFOT)  and average ergodic rate (AER) performance metrics, then maximize the AFOT. 
		In \cite{caching2}, Xu \emph{et al.} proposed a cache-enabled HetNets with limited backhaul and analyzed the successful content delivery probability, successful delivery rate as well as energy efficiency theoretically.}
	Caching those files proactively  during off-peak time at the edge of the network can reduce the data traffic pressure of backhaul link \cite{MostPop}}. When the backhaul traffic is offloaded by caching popular files at the cache of SBSs,  \textcolor{blue}{a part of mmWave spectrum can be transferred from the backhaul link to the access link. }

%
%
%
%

Although the introduction of cache brings the benifits of improvement the network throughput, nonetheless some new problems arise both theoretically and technically\textcolor{blue}{: how to deploy the cache appropriately to solve the \emph{``spectrum occupacy''} problem, how much improvement the cache can bring to the throughput performance of the mABHetNets. Thus, it is essential to jointly consider the cache decision and spectrum partition in mABHetNets to improve the APT.} In this paper,   average potential throughput (APT) metric is used to measure the network performance, which has become the major performance metric in mmWave HetNets\cite{Throuhgput1,Throuhgput2,Throuhgput3,ASE2} and   focuses on analyzing   user's average throughput with the specific rate requirement \cite{Throuhgput1}. 
On the one hand, since each  BS in cache-enabled mABHetNets has a limited energy resource and the energy consumed by additional cache can not be neglected\cite{BSPowerModel,BSPowerModel1},  caching capacity at SBS changes the signal-to-interference-plus-noise ratio (SINR) and coverage probability of user, which affects the corresponding spectrum partition strategy and further changes the allocation of resources(i.e., caching capacity, spectrum partion and power allocation) in mABHetNets. Therefore, in theory, APT in cache-enabled mABHetNets needs to \textcolor{blue}{consider the effect of the above factors.} 
On the other hand, \textcolor{blue}{caching capacity as well as transmit power are limited by the maximum power constraint at SBS, which also affects the optimal spectrum partition coefficient and thus affect the APT, so these three resource variables are coupled with each other.}
 To achieve desired APT in cache-enabled mABHetNets, the joint optimization algorithm needs to be further investigated.
 To address the aforementioned issues, in this paper, the main contributions are summarized as follows. 
\begin{itemize}
\item 
We develop a tractable analytical framework \textcolor{blue}{from a stochastic geometry perspective and derive the user association probability via wireless access link as well as SBS association probability via wireless backhaul link based on the maximum biased received power, respectively.}  \textcolor{blue}{Considering the effect of blockage and directional beamforming, we derive the distance distribution between the different transmitters and receivers by taking the association probability into account.} Then the APT expression is derived to measure the performance of cache-enabled mABHetNets. \textcolor{blue}{Under noise-limited scenario, a closed-form  SINR distribution in APT expression can be obtained through stochastic geometry tools to provide some insights about our proposed cache-enabled mABHetNets. }

\item
Based on our analytical work, an APT maximization problem is formulated where caching capacity, spectrum partition and power allocation are jointly considered;
As the formulated problem is a mixed-integer nonlinear programming (MINLP) problem, we decompose it into two sub-problems, i.e., cache decision problem, spectrum partition and power allocation problem. Then, inspired by block coordinate descent (BCD) method, we propose a joint cache decision, spectrum partition and power allocation (JCSPA) algorithm  to approach the optimal solution in an alternative manner;
\item 
We investigate the effects of caching capacity and spectrum partition on APT both theoretically and experimentally. Some important insights on the interplay between caching capacity and spectrum partition are provided from the perspective of APT increment. Numerical simulation results are carried out to verify the convergence of proposed algorithm and show the effects of other cache-related parameters on APT. These indicate that joint optimization of cache decision and spectrum partition is an effective method to bring \textcolor{blue}{about 90\%  increment} on APT beyond traditional mABHetNets. 
\end{itemize}

The rest of the paper is organized as follows. Firstly, we introduce the system model of \textcolor{blue}{cache-enabled} mABHetNets in Section II. Section III derives the SINR distribution of cache-enabled mABHetNets and then APT is further def\mbox{}ined and derived based on the SINR distribution. Next, Section IV. gives the APT maximization problem and solution.  Performance evaluation and numerical  results are provided in Section V. Finally, the paper is concluded in Section VI.

\section{System Model}
   \subsection{Network Model}
   In this section, we come up with a downlink cache-enabled  mABHetNet with \textcolor{blue}{integrated access and backhaul} architecture, which consists of an MBS tier and an SBS tier as shown in Fig. \ref{example1}.
In this architecture, high power MBSs are connected to core  network via broad high rate optical f\mbox{}iber links and SBSs are associated with the corresponding MBS via providing mmWave spectrum wireless backhaul transmission links. \textcolor{blue}{The typical user could be associated with both the MBS and SBS to obtain the wireless access service\cite{OCF}.}
    By stochastic geometry tool, the locations of the MBS and SBS are modeled as the independent Poisson Point Processes(PPPs), which are denoted by $\Phi_m\in \mathbb{R}_2$ and $\Phi_s\in \mathbb{R}_1$ with densities of $\lambda_m$ and $\lambda_s$, respectively. We stipulate that the user density is large suf\mbox{}f\mbox{}iciently so that each \textcolor{blue}{BS} consists of at least one associated user in its coverage area. \textcolor{blue}{We select a typical user at the origin for analysis. Based on Slivnyak's theorem in stochastic geometry, placing a point at the origin will not change the property of PPP.}

\begin{figure}[htbp]
	\centering
	\includegraphics[width=5in]{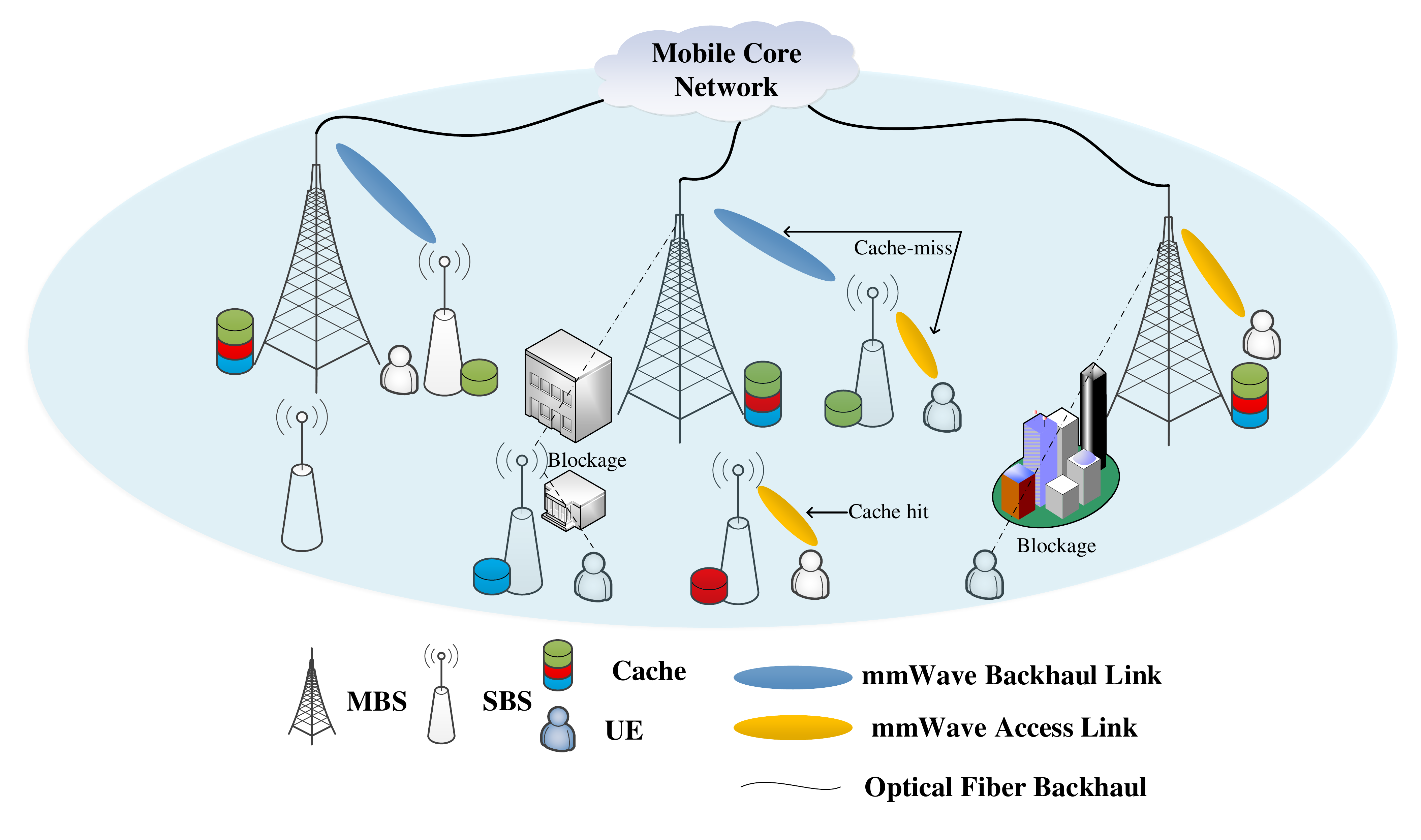}
	\captionsetup{name={Fig.},labelsep=period,singlelinecheck=off,font={small}}
	\caption{\textcolor{blue}{An illustration of cache-enabled millimeter wave heterogeneous network with integrated access and backhaul.} }
	\label{example1}
\end{figure}
%

\subsection{Caching Model}
The architecture we proposed in this paper is a cache-enabled HetNet. File library is denoted by symbol $\mathcal{F}$ and the number of all f\mbox{}iles is $|\mathcal{F}|=F$. In order to derive conveniently, we assume that all the f\mbox{}iles in library have the same size and their size is expressed in f\mbox{}ile units. Such assumption can hold since those files can be divided into chunks of equal size in practical transmission \cite{Samesize}. In the whole f\mbox{}ile library, different f\mbox{}iles have different popularities and the popularity of the f\mbox{}ile often changes slowly or doesn't change in a short time \cite{Popularity}. The popularity of the f\mbox{}ile can be predicted accrding to interview records or machine learning retrieval methods and this research problem has been actively studied in recent years while it is outside the scope of this paper. In general, the popularity of the f\mbox{}ile obeys the Zipf distribution\cite{BSPowerModel}\cite{Zipf0}. For the f\mbox{}iles denoted by the index set  $ \mathcal{F}={[1,2,...,f,...,F]}$, the popularity of the f\mbox{}ile $ f \in \mathcal{F} $ is
$p_{f}=\frac{f^{-\gamma_p}}{\sum_{g=1}^{F}g^{-\gamma_p}}$,
of which $\gamma_p$ is the distribution parameter of the Zipf distribution and it denotes the skewness of Zipf distribution, while the bigger $\gamma_p$ means that the fewer files have higher popularity\cite{Zipf1}. 
  In this paper, we use the highest-popularity-first cache strategy, which means that files with higher popularity will be cached preferentially. The caching capacity of MBS is large enough to load all the f\mbox{}iles in the f\mbox{}ile library\cite{MBSFileLibrary}, so the MBS caches all the $F$ f\mbox{}iles. The partial f\mbox{}iles are deployed in the SBS according to the level of the popularity, so the SBS caches the $C$ f\mbox{}iles in descending order of f\mbox{}ile popularity. From these  we can also see that the MBS is `Full-Cache' and the SBS is `Partial-Cache'. Based on this, the cache hit ratio of the SBS is
\begin{equation}\label{HitRatio}
p_{h}(C)= \frac{\sum_{f=1}^{C}f^{-\gamma_p}}{\sum_{g=1}^{F}g^{-\gamma_p}}.
\end{equation}

\subsection{Power Consumption Model}

Like \cite{networkPowerModel}, the total power consumption model of the mABHetNet can be regarded as the power consumption of the MBS and SBS, where the power consumption of the MBS is $\rho_m P_{m}^{\mathrm{tr}}+P_{m}^{\mathrm{fc}}+P_{m}^{\mathrm{ca}}=\rho_{m} P_{m}^{\mathrm{tr}}+P_{m}^{\mathrm{fc}}+\omega_{ca}\textcolor{blue}{s}F$ and the power consumption of the SBS is $\rho_s P_{s}^{\mathrm{tr}}+P_{s}^{\mathrm{fc}}+P_{s}^{\mathrm{ca}}=\rho_s P_{s}^{\mathrm{tr}}+P_{s}^{\mathrm{fc}}+\omega_{ca}\textcolor{blue}{s}C$. $\rho_m$ and $\rho_s$ are the power consumption amplif\mbox{}ication factor of MBS and SBS transmitter. $P_{m}^{\mathrm{fc}}$ and $P_{s}^{\mathrm{fc}}$ are the f\mbox{}ixed circuits-related power consumption of MBS and SBS\cite{Powercc}. 
$P_{m}^{\mathrm{ca}}$ and $P_{s}^{\mathrm{ca}}$ are the caching power consumption of MBS and SBS based on energy proportional model. In order to quantify the caching power consumption, widely used proportional power model \cite{BSPowerModel,BSPowerModel0},\cite{Powercc,frapca1,frapca2} is introduced in this paper so that the caching power consumption is proportional to the cache f\mbox{}ile size such as $P_{m}^{\mathrm{ca}}=\omega_{ca}\textcolor{blue}{s}F$, $P_{s}^{\mathrm{ca}}= \omega_{ca}\textcolor{blue}{s}C$, of which $\omega_{ca}$ is the coef\mbox{}f\mbox{}icient of cache power consumption(unit:watt/bit). Assuming that the maximum power of each SBS has been preset, each SBS can adjust the transmit power and caching power(caching capacity) without exceeding the preset maximum power limit\cite{BSPowerModel0}.
The actual transmit power consumption of the SBS and MBS are  $P_s^{\mathrm{tr}}=\frac{P_s^{\mathrm{max}}-P_s^{\mathrm{fc}}-\omega_{ca}\textcolor{blue}{s}C}{\rho_s}=P_s'- \omega_{ca}^{'s}C$,  $P_m^{\mathrm{tr}}=\frac{P_m^{\mathrm{max}}-P_m^{\mathrm{fc}}-\omega_{ca}\textcolor{blue}{s}F}{\rho_m}=P_m'- \omega_{ca}^{'m}F$, \textcolor{blue}{of which $P_s^{'}=\frac{P_s^{\mathrm{max}}-P_s^{\mathrm{fc}}}{\rho_s}$, $\omega_{ca}^{'s}=\frac{\omega_{ca}{s}}{\rho_s}$, $P_m^{'}=\frac{P_m^{\mathrm{max}}-P_m^{\mathrm{fc}}}{\rho_m}$ and $\omega_{ca}^{'m}=\frac{\omega_{ca}{s}}{\rho_m}$.}
The relevant notations are summarized in the Table \ref{Symbols}.

\begin{table}[htbp]
	\centering
	\footnotesize
	\setlength\tabcolsep{3pt}
	\caption{Main Notations}
	\label{Symbols}
	\renewcommand\arraystretch{0.8}
	\begin{tabular}{|l|l|l|l|}
		\hline
		\textbf{Notation} & \textbf{Meaning} & \textbf{Notation} & \textbf{Meaning} \\ \hline
		$W/W_{ac}/W_{bh}$ &Total spectrum/access link/backhaul link bandwidth    &$\lambda_m/\lambda_s$       &Density of MBS/SBS        \\ \hline
		$\alpha_\mathrm{L}/\alpha_\mathrm{NL}$&Path loss exponet of LOS/NLOS transmission         &$C$           &Caching capacity of SBS          \\ \hline
		 $\eta$         &mmWave bandwidth partition ratio for access link          & $C_{\mathrm{max}}$    &Maximum caching capacity of SBS   \\ \hline
		$P_{m}^{\max}/P_{s}^{\max}$         &Maximum power of MBS/SBS         & $F$               &Number of f\mbox{}iles in file library       \\ \hline
		$P_{m}^{\mathrm{tr}}/P_{s}^{\mathrm{tr}}$             &Transmit power of MBS/SBS         &  $p_h$  &Cache hit ratio of  SBS      \\ \hline
		{\color{blue}$P_{m}^{\mathrm{fc}}/ P_{s}^{\mathrm{fc}}$}  &\textcolor{blue}{Fixed circuits-related power consumption of MBS/SBS }   & {\color{blue}$\theta$}      &\textcolor{blue}{Mainlobe beamwidth}  \\ \hline
		{\color{blue}$P_{m}^{\mathrm{ca}}/P_{s}^{\mathrm{ca}}$} & \textcolor{blue}{Caching  power consumption of MBS/SBS}     &  $B_m/B_s$                          &The association bias factor of MBS/SBS       \\ \hline
		{\color{blue}$\rho_m/\rho_s$}   &\textcolor{blue}{Power consumption amplification factor of MBS/SBS}      &    {\color{blue}$s$}                      &\textcolor{blue}{Size of a file} \\ \hline
		{\color{blue}$M/m$}                   &\textcolor{blue}{The antenna gain of mainlobe/slide lobe}      &    $\omega_{ca}$                      &Caching power consumption coef\mbox{}f\mbox{}icient        \\ \hline	
	\end{tabular}
\end{table}
\subsection{Channel and Transmission Model}
In this paper, on account of the high density of  MBS or SBS deployment and the linear propagation transmission characteristic of mmWave, each transmission link is assumed to be indepedent  Nakagami-Rayleigh fading and the signals are based on the transmission link as either line-of-sight(LOS) or non-line-of-sight(NLOS). The path loss function expresses the signal attenuation relationship with distance $r$. The mathematical expression of the path loss function is as follows\cite{LOSModel}:

\begin{align}\label{LOSModel}
	\operatorname{L}(r)
	=\left\{
	\begin{array}{ll}
		{A_{\mathrm{L}} r^{-\alpha_{\mathrm{L}}},}      &{\text{with LOS probability $\mathcal{P}_\mathrm{L}(r)$}}, \\
		{A_{\mathrm{NL}} r^{-\alpha_{\mathrm{NL}}},}    &{\text{with NLOS probability $\mathcal{P}_\mathrm{NL}(r)=1-\mathcal{P}_\mathrm{L}(r)$}},
	\end{array}
	\right.
\end{align}
of which the  LOS probability  $\mathcal{P}_{\mathrm{L}}(r)=e^{-\beta r}$
, $A_{\mathrm{L}}$($A_{\mathrm{NL}}$) is the LOS(NLOS) pathloss parameter, $r$ is the distance between the user and base station, $\alpha_{\mathrm{L}}$($\alpha_{\mathrm{NL}}$) is a path loss exponent in LOS(NLOS) transmission of Nakagami fading. $\beta \geq 0$ is the parameter that captures
density and size of obstacles between the transmitter and the reveicer. 
In practice, the probability of LOS transmission coverage is very close to one and  NLOS  transmission  could be neglected.

\textcolor{blue}{Besides, directional beamforming  are used for all antennas at the transceivers and signals propagates along the main lobe of the antenna. We use the sectorial antenna pattern in this analysis \cite{beamforming}. The antenna gain pattern for the transceiver in the mABHetNets is given as 
\begin{equation}
	G_q(\phi)= \left\{
	\begin{aligned}
		&M , \qquad \text{if} \left |\phi \right| \le\theta \\
		&m , \qquad \text{otherwise}.
	\end{aligned}
	\right.
\end{equation}	
where $q\in \{T,R\}$ denotes the antenna at the transmitter or receiver, $\phi \in [0,2\pi)$ is the angle off  boresight direction, $\theta$ is the mainwidth of mainlobe,  $M$ and $m$ are the gain of the main lobe and slide lode.  Then the random gain for the transmission link and its probability is given as
\begin{equation}
	G= \left\{
	\begin{aligned}
		&MM , \qquad \text{with probability}  \frac{\theta^2}{4\pi^2},  \\
		&Mm , \qquad \text{with probability} \frac{\theta(2\pi-\theta)}{2\pi^2},\\
		&mm , \qquad \text{with probability} \frac{(2\pi-\theta)^2}{4\pi^2}.
	\end{aligned}
	\right.
\end{equation}
For tractbility of analysis, the perfect beamforming is assumed between the transmitter and receiver \cite{beamforming1}.}
Based on the above analysis, we can derive the the signal to interference plus noise ratio expression of a typical user from the associated SBS or the associated MBS \textcolor{blue}{via wireless access link} at the distance $r$ as follows.
\begin{small}
\begin{align}\label{SBSUSERSINR}
		&\operatorname{SINR}_{s}(r_s)
		=\frac{P_{s}^{tr} B_{s}\textcolor{blue}{G_{s}}h_{s  } L(r_s)}{I_s+I_m+N_{0}} 
		=\frac{(P'_s-\omega^{'s}_{\mathrm{ca}}C) B_{s}\textcolor{blue}{G_{s}}h_{s  } L(r_{s })}
		{\sum\limits_{i \in \Phi_{s} \backslash b_{s,0}} (P'_{s}-\omega^{'s}_{\mathrm{ca}}C) B_{s }\textcolor{blue}{G_{i}} h_{s, i}L(r_{s, i})+
			\sum\limits_{l \in \Phi_{m} } P_{m}^{tr} B_{m }\textcolor{blue}{G_{l}} h_{m, l}L(r_{m, l})+N_{0}},
\end{align}
\end{small}
\begin{small}
\textcolor{blue}{\begin{align}\label{MBSUSERSINR}
		&\operatorname{SINR}_{m}(r_m)
		=\frac{P_{m}^{tr} B_{m}\textcolor{blue}{G_{m}}h_{m } L(r_m)}{I'_s+I'_m+N_{0}} 
		=\frac{(P'_m-\omega^{'m}_{\mathrm{ca}}F) B_{m}\textcolor{blue}{G_{m}}h_{m} L(r_{m})}
		{\sum\limits_{i \in \Phi_{s} } (P'_{s}-\omega^{'s}_{\mathrm{ca}}C) B_{s }\textcolor{blue}{G_{i}} h_{s, i}L(r_{s, i})+
			\sum\limits_{l \in \Phi_{m}\backslash b_{m,0} } P_{m}^{tr} B_{m }\textcolor{blue}{G_{l}} h_{m, l}L(r_{m, l})+N_{0}},
\end{align}}
\end{small}
where $b_{s,0}$ and \textcolor{blue}{$b_{m,0}$} denotes the serving SBS \textcolor{blue}{and MBS} for user, $B_s$ and $B_m$ are the association bias factor of SBS and MBS, \textcolor{blue}{$G_{m}$, $G_{s}$, $G_{i}$ and $G_{l}$ are the antenna gain of the correspoding transmission link($G_{m}=G_{s}=MM$).} $h_{s,i} $ \textcolor{blue}{and $h_{m,l}$ are}  the small-scale fadings from $i$-th SBS or \textcolor{blue}{$j$-th SBS} ($h_s\textcolor{blue}{=h_m}$ $\sim$ $\exp(1)$).  $L(r_{s})$ \textcolor{blue}{$L(r_{m})$} are the path loss from the serving SBS or MBS to the typical user. $r_{s}$  is the distance between the association SBS $b_{s,0}$ and the typical user. 
\textcolor{blue}{$r_{m}$  is the distance between the association MBS $b_{m,0}$ and the typical user.} 
$L(r_{s, i})$ is path loss from the $i$-th SBS to the typical user.
 \textcolor{blue}{$L(r_{m, l})$ is the path loss from the $l$-th MBS to the typical user.}
 $N_{0}$ is the additive white Gaussian noise component.
Besides, the MBS provides the wireless backhaul link to the SBS for data transmission. Similar to (\ref{SBSUSERSINR}) \textcolor{blue}{and (\ref{MBSUSERSINR})}, for a typical SBS at a random distance $r_{bh}$ from its associated MBS,  the SINR of the signal from the MBS to the SBS of the downlink backhaul link is then given as follows.
\begin{small}
\begin{align}\label{M-SBSSINR}
		&\operatorname{SINR}_{bh} (r_{bh})
		=\frac{P_{m}^{tr} B_m\textcolor{blue}{G_{m}} h_{m } L(r_{bh})}{ I_{bh}+  N_0}
		=\frac{(P'_m-\omega^{'m}_{\mathrm{ca}}F) B_m\textcolor{blue}{G_{m}} h_{m } L(r_{bh})}
		{ \sum\limits_{i \in \Phi_{m} \backslash b_{m,0 }}     P_{m}^{tr} B_m \textcolor{blue}{G_{i}} h_{m,i}L(r_{bh,i})+N_0}.
\end{align}
\textcolor{blue}{The symbol have similar meanings with (\ref{SBSUSERSINR}) and (\ref{MBSUSERSINR}).}
\end{small}
\subsection{Spectrum Partition Strategy}
  In this section, we will introduce the spectrum partition strategies in the downlink mABHetNet. $W$ is the whole millimeter wave spectrum bandwidth available for the mABHetNet.  To avoid the spectrum interference,  the spectrum of access and backhaul link is orthogonal, so the whole spectrum  needs to be divided into two parts: $W_{bh}$ for the backhaul link and $W_{ac}$  for the access link. Then we will introduce two bandwidth allocation strategies in the following.
  \subsubsection{Fixed Spectrum Allocation(FSA)}This default spectrum allocation strategy  is  widely used in traditional heterogeneous network. In order to satisfy the communication needs failrly, we allocate  spectrum for the access and backhaul link  equally as $W_{bh}= 0.5W$ and $W_{ac}= 0.5W$($W$ is the total spectrum bandwidth). In this case, we will give the transmission data rate based on Shannon's theorem of backhaul link and access link, respectively.
  \begin{align}
  	R_{i}= \frac{1}{2}W\operatorname{log}_{2}(1+\operatorname{SINR}_{i}),
  \end{align} 
where $i=\{bh, m,s \} $ denotes the backhaul link, MBS tier or SBS tier via access link.
\subsubsection{Dynamic Spectrum Allocation(DSA)}
  In order to alleviate  ``\emph{spectrum occupancy}'' problem, we introduce a parameter $\eta$ to allocate spectrum dynamically according to the current caching status at SBS, where $\eta\in\left[0,1\right]$ is the spectrum partition ratio coef\mbox{}f\mbox{}icient for the access link so the  spectrum for the  backhaul link is $W_{bh}= (1-\eta)W$ and the spectrum for the access  link is $W_{ac}= \eta W$. Users are usually accessed to the SBSs and f\mbox{}ile delivery needs to go through the backhaul link, so the access link  is cache-miss when the f\mbox{}iles requested at not cached at the SBS. In other words, the f\mbox{}ile delivery rate depends on both the access link rate and the backhaul link rate. Similarly, we give the transmission data rate of each link as follows.
  \begin{figure}[htbp]
  	\centering
  	\includegraphics[width=5in]{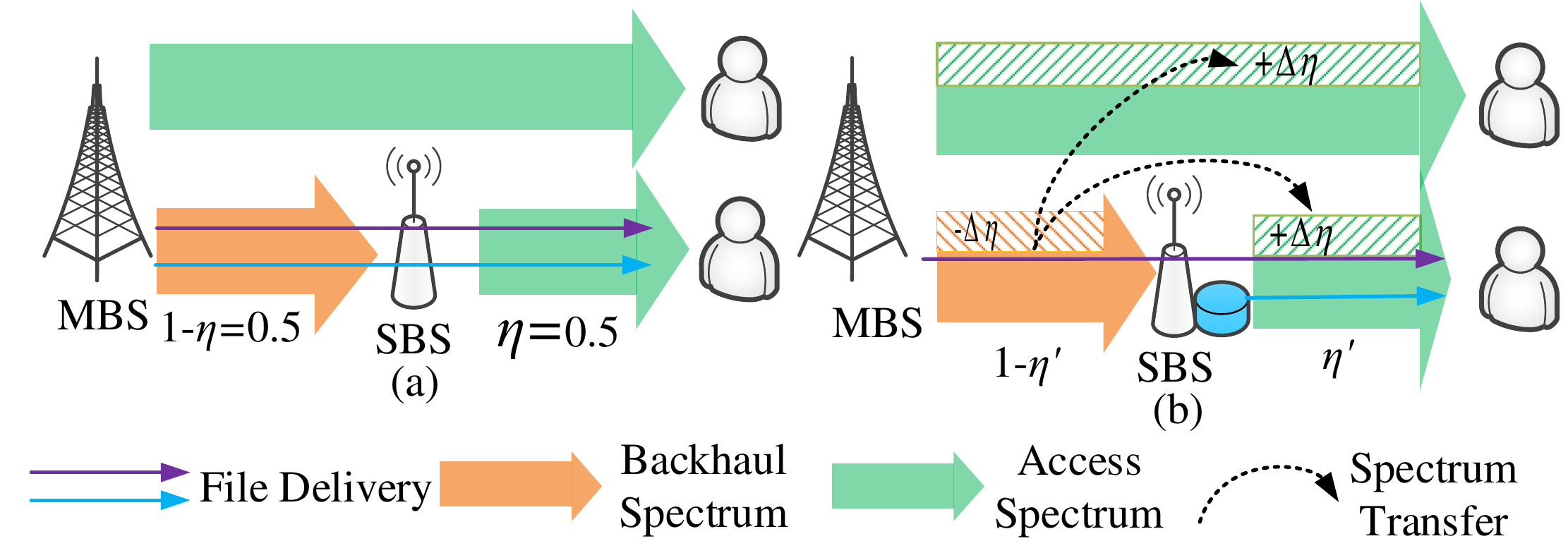}
  	\captionsetup{name={Fig.},labelsep=period,singlelinecheck=off,font={small}}
  	\caption{Spectrum partition strategies between access and backhaul link, (a) Fixed Spectrum Allocation(FSA), (b) Dynamic Spectrum Allocation(DSA). }
  	\label{schematic}
  \end{figure}
  \begin{align}
  	&R_{bh} = (1-\eta)W\operatorname{log}_{2}(1+\operatorname{SINR}_{bh}),\\
  	&R_{s} = \eta W\operatorname{log}_{2}(1+\operatorname{SINR}_{s}),\\
  	&\textcolor{blue}{R_{m} = \eta W\operatorname{log}_{2}(1+\operatorname{SINR}_{m}).}
  \end{align}
\section{SINR Distribution and APT Analysis  of mABHetNet}
In this section, we will derive the expression of SINR distribution of the typical user conditioned on its association selections. As users are covered  by the associated SBS,  we f\mbox{}irst derive the PDF of the distance   between the user and  the serving SBS. Further,  SINR distributions of the SBS associated with the serving MBS are also obtained. Finally, we derive the APT expression of cache-enabled mABHetNets.
\subsection{The PDF of Distance to Nearest Base Station}
First of all, we need to derive the probability distribution function (PDF) of the distance between the typical  user and its nearest SBS. We focus on the  typical user at the origin.  When the typical user communicates with the closest SBS at a distance $r$, no other SBS can be closer than $r$. 
Since the typical user is associated with the closest SBS via either LOS or NLOS channel, we derive these PDFs in the following Lemma.
\begin{lemma}\label{LemmaUserClosestSBSProbabilty}
The PDF of $r$ (the distance between the typical user and the  nearest SBS via a LOS/NLOS path) is written as
    \begin{align}\label{LOSUserClosestSBSProbabilty}
        f_{R_{s}}^{\mathrm{k}}(r)     &=\mathcal{P}_{k}(r)  \times \exp \left(-\pi r^{2} \lambda_s\right)\times 2 \pi r \lambda_s,
    \end{align}
where $s $ denotes the index of SBS tier, $\mathrm{k} =\{\mathrm{L}, \mathrm{NL}\} $ denotes the  transmission path of LOS or NLOS  in wireless access link.
\end{lemma}
\begin{proof}
The proof process is provided in Apendix \ref{Appe1:NearDistance}.
\end{proof}
\begin{remark}
  Similarly to Lemma \ref{LemmaUserClosestSBSProbabilty}, the uncached f\mbox{}ile data will be delivered from the MBS to the SBS by the wireless backhaaul link of mmWave. Similar to the above analysis in Lemma \ref{LemmaUserClosestSBSProbabilty}, the PDF  of distance $r$ (between  the SBS  and the  associated the nearest MBS via a LOS/NLOS path)  can be written as
\begin{align}\label{PDFSBSMBSLOS}
  f_{R_{bh}}^{\mathrm{k}}(r)&=\mathcal{P}_{k}(r) \times  \exp \left(-\pi r^{2} \lambda_m\right)\times 2 \pi r \lambda_m,
\end{align}
  where bh denotes the index of backhaul link and $\mathrm{k}$ is the same as (\ref{LOSUserClosestSBSProbabilty}). 
\end{remark}
\subsection{Association Probability}
In the mABHetNet, we will f\mbox{}irst analyze the probability that a user  is associated with SBS tier via different path due to the different  transmission path 
and the densities of the SBS and MBS.  
Besides, since SBS may be associated with MBS via different transmission paths in the SBS backhaul association, different SBS backhaul association probabilities should be derived in the meanwhile.
\subsubsection{ user association probability}
We consider a user association based on maximum biased received power, where a mobile user is associated with the strongest SBS in terms of  the received power at the user. Then, for  a typical user associated with the SBS tier via LOS path and NLOS path, the received powers  are $P_{s}^{tr}  B_s \textcolor{blue}{G_{s}}h_{s}  A_\mathrm{L} r^{-\alpha_\mathrm{L}}$ and $P_{s}^{tr}  B_s\textcolor{blue}{G_{s}}h_{s}  A_\mathrm{NL} r^{-\alpha_\mathrm{NL}}$, respectively. 
Based on the maximum biased received power association strategy, the BS density and transmit power as well as transmission path determine the probability that a typical user is associated with an SBS. The following lemma provides the SBS association probability via LOS and NLOS path respectively.
\begin{lemma}\label{LemmaUserAssociatedSBS}
   For the given distance $r$, the probabilities that a typical user is associated with the SBS tier by LOS link and NLOS link are
\begin{align}
    F_{s}^{\mathrm{L}}(r) &=p_{ln}^{ss}(r)\textcolor{blue}{p_{ll}^{sm}(r)p_{ln}^{sm}(r)}f_{R_s}^{\mathrm{L}}(r),\label{UserAssociatedSBS0}\\
    F_{s}^{\mathrm{NL}}(r)&=p_{nl}^{ss}(r)\textcolor{blue}{p_{nl}^{sm}(r)p_{nn}^{sm}(r)}f_{R_s}^{\mathrm{NL}}(r),\label{UserAssociatedSBS1}
\end{align}
\textcolor{blue}{
Then, the probabilities that a typical user is associated with the MBS tier by LoS link and NLoS link are
\begin{align}\label{UserAssociatedMBS}
	F_{m}^{\mathrm{L}}(r)&=p_{ln}^{mm}(r)p_{ll}^{ms}(r)p_{ln}^{ms}(r)f_{R_{m}}^{\mathrm{L}}(r),\\
	F_{m}^{\mathrm{NL}}(r)&=p_{nl}^{mm}(r)p_{nl}^{ms}(r)p_{nn}^{ms}(r)f_{R_{m}}^{\mathrm{NL}}(r)
\end{align}}
where 
$p_{ln}^{ss}(r)=e^{-\lambda_s\pi\left[ \left(\frac{A_\mathrm{L}}{A_\mathrm{NL}}\right)^{\frac{-1}{\alpha_{\mathrm{NL}}}} r^{\frac{\alpha_{\mathrm{L}}}{\alpha_{\mathrm{NL}}}}\right]^{2}}$ denotes the probability of the event that the user obtains the desired LOS signal from SBS and the NLOS  interference from the SBS. Other notations have  similar def\mbox{}inition and can be found in the proof.
\end{lemma}
\begin{proof}
The proof process is provided in Apendix \ref{Appen2:AssoPro}.
\end{proof}

\subsubsection{SBS backhaul association probability}
The SBS will be associated with the MBS by the wireless backhaul link. The SBS backhaul association strategy is also based on the maximum biased received power from the MBS. Since the backhaul transmission includes LOS link and NLOS link, there exist two   backhaul association probabilities. These  probabilities are given in the following remark.
\begin{remark}
Similar to Lemma \ref{LemmaUserAssociatedSBS}, the probabilities that a typical SBS is associated with   the MBS tier by LOS link and NLOS
link are
\begin{align}\label{SBSAssociatedMBS}
F_{bh}^{\mathrm{L}}(r)=p_{ln}^{bh}(r) f_{R_{bh}}^{\mathrm{L}}(r),\\ 
F_{bh}^{\mathrm{NL}}(r)=p_{nl}^{bh}(r) f_{R_{bh}}^{\mathrm{NL}}(r),
\end{align}
where
$p_{ln}^{bh}(r)=e^{-\lambda_m\pi\left[\left(\frac{A_\mathrm{L}}{A_\mathrm{NL}}\right)^{\frac{-1}{\alpha_{NL}}} r^{\frac{\alpha_{L}}{\alpha_{NL}}}\right]^{2}}$ is the probability that the SBS is associated with the  MBS in LOS link and
the interference is from  MBS in NLOS link.  $p_{nl}^{bh}(r)$ has the similar definition.
\end{remark}
\subsection{SINR Distribution}
To study the APT  performance of mABHetNet, we need to f\mbox{}irst  investigate the SINR distribution of the user  covered by SBS tier via access link and the SINR distribution of  the SBS   covered by MBS via backhaul link.
This SINR distribution is def\mbox{}ined as the  probability that the received SINR is above a pre-designated threshold $\gamma$:
\begin{equation}
P^{\operatorname{cov}}(  \gamma)=\operatorname{Pr}[\operatorname{SINR}>\gamma].
\end{equation}
Since the user is covered by SBS, we f\mbox{}irst give the  SINR distribution of the user. Then we give the SINR distribution of the typical SBS which is covered by MBS.
\begin{proposition}\label{ProUserSINRCoverageProbability}
1) SINR distribution of user covered by the SBS \textcolor{blue}{or the MBS} :

The SINR distribution that the typical user is associated with the SBS via access link is:
\begin{align}\label{UserSBSSINRCoverageProbability}
\mathbb{P}_{s}^{cov}(\gamma)=\sum_{i=\mathrm{L,NL}}P_{s,i}^{cov}(\gamma),
\end{align}
\begin{align}
P_{s,i}^{cov}(\gamma)=
\int_{0}^{\infty}  \exp \left(\frac{-\gamma N_{0}}{P_{s}^{tr}B_s \textcolor{blue}{G_s} A_{\mathrm{i}} r^{-\alpha_{\mathrm{i}}} }\right)  \mathcal{L}_{I_{s\textcolor{blue}{,m}}}^{\mathrm{i}}   F_{s}^{\mathrm{i}}(r)dr,\label{UserSBSSINRCoverageProbabilityNLoS}
\end{align}
of which s denotes the SBS, $i=\{\mathrm{L,NL}\}$ denotes the transmission link. \textcolor{blue}{$G_s=MM$}, $P_{s,i}^{cov}(\gamma)=\mathbb{E}_{r }\left[\mathbb{P}\left[\mathrm{SINR}_{s}^{\mathrm{i}} (r ) \geq \gamma\right]\right]$ is the probability that the user is covered by the SBS, where $s$ denotes SBS,  $\gamma$ is the default threshold for successful demodulation and decoding at the receiver and the Laplace transform of interference from SBS or MBS in LOS link is given as
\begin{small}
	\begin{align}
		\mathcal{L}_{I_{s,m}}^{\mathrm{L}}\left( \gamma r^{\alpha_{\mathrm{L}}} \right)
		&\stackrel{(b)}{=}
		\textcolor{blue}{\prod \limits_{G_i}}\exp \left(-2 \pi  \lambda_{s} \textcolor{blue}{p_{G_i}}\left( \int_{r}^{\infty}  \frac{\mathcal{P}_L(u)u}{1+\frac{P_s^{tr}B_{s}\textcolor{blue}{G_{s}} A_Lr^{-\alpha_{\mathrm{L}}}}{\gamma P_s^{tr}B_{s}\textcolor{blue}{G_{i}}A_Lu^{-\alpha_{\mathrm{L}}}}} d u+\int_{\left(\frac{A^{\mathrm{L}}}{A^{\mathrm{NL}}}\right)^{\frac{-1}{\alpha^{\mathrm{NL}}}} r^{\frac{\alpha_{\mathrm{L}}}{\alpha_{\mathrm{NL}}}}}^{\infty} 
		\frac{\mathcal{P}_{\mathrm{NL}}(u)u}{1+\frac{P_s^{tr}B_{s}\textcolor{blue}{G_{s}}A_{\mathrm{L}}r^{-\alpha_{\mathrm{L}}}}{\gamma  P_s^{tr}B_{s}\textcolor{blue}{G_{i}} A_{\mathrm{NL}}u^{-\alpha_{\mathrm{NL}}}}} d u
		\right)\right) \nonumber\\
		&\times \textcolor{blue}{\prod \limits_{G_l}}\exp \left(-2 \pi  \lambda_{m}\textcolor{blue}{p_{G_l}}\left(\int_{\left(d_1\right)^{\frac{-1}{\alpha_{L}}} r}^{\infty} \frac{\mathcal{P}_{L}(u)u}{1+\frac{P_{s}^{tr}B_{s}\textcolor{blue}{G_{s}}A_L r^{-\alpha_{\mathrm{L}}}}{\gamma P_{m}^{tr} B_{m}\textcolor{blue}{G_{l}} A_L u^{-\alpha_{\mathrm{L}}}}} d u +\int_{\left(d_2\right)^{\frac{-1}{\alpha_{\mathrm{NL}}}} r^\frac{\alpha_{\mathrm{L}}}{\alpha_{\mathrm{NL}}}}^{\infty}  \frac{\mathcal{P}_{\mathrm{NL}}(u)u}{1+\frac{P_{s}^{tr} B_{s}\textcolor{blue}{G_{s}}A_L r^{-\alpha_{\mathrm{L}}}}{\gamma P_{m}^{tr}B_{m}\textcolor{blue}{G_{l}}A_\mathrm{NL} u^{-\alpha_{\mathrm{NL}}}}} d u\right) \right) \nonumber,
	\end{align}
\end{small}
where \textcolor{blue}{$G_i,G_l \in\{MM,Mm,mm\},p_{G_i}(p_{G_l})$ is the probability of the antenna gain taking correspoding value from SBS interference tier and MBS interference tier.} $d_1=\frac{P_s^{tr}B_{s} }{P_m^{tr}B_{m}}$ and $d_2=\frac{P_s^{tr} B_{s}A_\mathrm{L}}{P_m^{tr}B_{m}A_\mathrm{NL}}$. Following the same logic, other Laplace transform of interference have similar expressions.
\textcolor{blue}{
\begin{remark}
	From the Laplace transform of interference, the interference from the SBS tier is independent of the transmit power $P_s^{tr}$ and even the interference from the MBS layer is monotonically decreasing about $P_s^{tr}$, which also reveals that  increasing the transmit power $P_s^{tr}$ is benificial to improve the coverage probability and APT without considering other constraints. Such proof precess can be done by substituting the interference term into the Laplace transform.
\end{remark}}

%
\textcolor{blue}{
Similarly, the SINR coverage probability that the user is associated with the MBS is
\begin{align}\label{UserMBSSINRCoverageProbability}
	\mathbb{P}_{m}^{cov}(\gamma)=\sum_{i=\mathrm{L,NL}}P_{m,i}^{cov}(\gamma),
\end{align}
\begin{align}\label{UserMBSSINRCoverageProbabilityLoS}
	P_{m,i}^{cov}(\gamma)=\int_{0}^{\infty} \exp \left(\frac{-\gamma N_{0}}{P_{m}^{tr} B_m G_m A_{\mathrm{i}} r^{-\alpha_{\mathrm{i}}} }\right)  \mathcal{L}_{I'_{s,m}}^{\mathrm{i}}\left( \gamma r^{\alpha_i}\right) F_{m}^{\mathrm{i}}(r)d r
\end{align}
}


2) SINR distribution of SBS covered by  MBS:

The SINR distribution that the typical SBS is covered by the  MBS via backhaul link is:
\begin{align}\label{BHSINRCoverageProbability}
  \mathbb{P}_{bh}^{cov}(\gamma)   =\sum_{i=\mathrm{L,NL}}P_{bh,\mathrm{i}}^{cov}(\gamma),
\end{align}
\begin{align}\label{BHSINRCoverageProbabilityLoS}
  P_{bh,\mathrm{i}}^{cov}(\gamma)=\int_{0}^{\infty} \exp \left(\frac{-\gamma N_{0}}{P_{m}^{tr}  B_m \textcolor{blue}{G_m} A_{\mathrm{i}} r^{-\alpha_{\mathrm{i}}} }\right)  \mathcal{L}_{I_{bh}}^{\mathrm{i}}\left(\gamma r^{-\alpha_i}\right) F_{bh}^{\mathrm{i}}(r)dr,
\end{align}
where the subscript $bh$ denotes the backhaul, $i$ also represents the transmission link and \textcolor{blue}{$G_m=MM$}. Note that,  considering a more general fading model such as Nakagami does not provide any additional design insights, but it does complicate the analysis signif\mbox{}icantly. Similar to \cite{Rayleigh} in our paper,   the special case of Rayleigh fading is considered .
\end{proposition}
\begin{proof}
  The proof process is provided in Appendix \ref{Appen3:ProUserSINRCoverageProbability}.
\end{proof}
\textcolor{blue}{
\begin{proposition} \label{noise_limited}
	In noise-limited scenario ($\sigma^{2}>>I$), the SINR distribution can be reduced into a closed-form as
	\begin{equation}
		P_{k}^{cov}(\gamma) =1-\exp\left(-\pi \lambda_k A_\mathrm{NL}^{\frac{2}{\alpha_{\mathrm{NL}}}}\Gamma\left(\frac{1}{\alpha_\mathrm{NL}}+1\right)\left(\frac{P_kB_k G_k}{\sigma^{2}\gamma}\right)^{\frac{2}{\alpha_\mathrm{NL}}}-2 \pi \lambda_kY\left(\frac{P_kB_k G_k}{\sigma^{2}\gamma}\right)\right),
	\end{equation}
	where $Y(\xi_0)=\int_{0}^{\xi_0}\frac{A_\mathrm{L}^{\frac{2}{\alpha_\mathrm{L}}-1}}{ \alpha_\mathrm{L} \xi^2}\int_{0}^{\infty}\phi^{\frac{2}{\alpha_\mathrm{L}}}\exp(-\beta\phi^{\frac{1}{\alpha_\mathrm{L}}}-\frac{\phi}{\xi}) d\phi d\xi-\int_{0}^{\xi_0}\frac{A_\mathrm{NL}^{\frac{2}{\alpha_\mathrm{NL}}-1}}{ \alpha_\mathrm{NL} \xi^2}\int_{0}^{\infty}\phi^{\frac{2}{\alpha_\mathrm{NL}}}\exp(-\beta\phi^{\frac{1}{\alpha_\mathrm{NL}}}-\frac{\phi}{\xi}) d\phi d\xi
	$,  $k \in \{m,s\}$ denotes the MBS tier or the SBS tier.
\end{proposition}
\begin{proof}
	The detailed proof process of  Proposition\ref{noise_limited} is provided in the Appendix.
\end{proof}}

%

\subsection{ APT of Cache-enable mABHetNet}
APT is  a signif\mbox{}icant metric to measure the heterogeneous network performance, which mainly focuses on the average user QoS requirement in terms of data rate. Next, we will derive the APT expressions based on the above analysis.
\subsubsection{Def\mbox{}inition}
APT captures the average number of bits that can be received by user per unit area per unit bandwidth for a given pre-designated threshold $\gamma_0$ \cite{Throuhgput1}. The def\mbox{}inition of APT is 
\begin{equation}
   \mathcal{R}\left(  \gamma_0\right) =\lambda_k W\log _{2}\left(1+\gamma\right) \mathbb{P}\left\{\operatorname{SINR} \geq \gamma_0\right\},
\end{equation}
where $\lambda_k,k\in \{m,s\}$ is the density of MBS or SBS and $W$ is the  allocated spectrum. $\gamma_0$ is SINR threshold for the signal demodulation and represents the receiver's communication requirement.

\subsubsection{The detailed APT expression}
 APT of mABHetNet depends on both the backhaul link and the access link, then we will give the detailed expression  in the following analysis.

For the user associated with the SBS, the transmission link includes the backhaul link bewteen the MBS and SBS and the access link between the SBS and user. Besides, in this cache-enabled mABHetNet, the cache in SBS also influence the f\mbox{}ile delivery in the transmission path. When the requested f\mbox{}iles  are cached at the SBS,  the f\mbox{}iles  can be delivered to user directly without transmitting through the wireless backhaul link. In other words, the wireless backhaul link will not be occupied under cache-hit circumstance. Besides, if the f\mbox{}iles requested by user are not cached at the SBS, the missing f\mbox{}iles need to be delivered from the MBS, which increases the wireless backhaul traff\mbox{}ic. In this case, the user's APT depends on the minimum throughput of the backhaul link and the access link. \textcolor{blue}{For the user associated with MBS, APT depends on the access link.} Given that the partial f\mbox{}iles are  cached in SBS,   we  give the  APT expression of mABHetNet in the following corollary.
\begin{corollary}\label{corolAPTSBS}
Since the transmission  can be LOS or NLOS in wireless access  and  backhaul link for user associated with SBS tier \textcolor{blue}{or MBS tier, we can get the APT of SBS tier, MBS tier and mABHetNet respectively. First, we give the APT expression of SBS tier.} 
  \begin{align} \label{APTSBS}
  	&\mathcal{R}_{s}(\eta,C,P_{\mathrm{s}}^{\mathrm{tr}})  = \min \{ (1-p_h(C))\lambda_s\eta W\log_2(1+\gamma_0)P^{cov}_{s}(P_{\mathrm{s}}^{\mathrm{tr}}),\nonumber   \\ 
  	&(1-p_h(C)) \lambda_m W(1-\eta)\log_2(1+\gamma_0)\textcolor{blue}{P^{cov}_{bh}}(P_{\mathrm{m}}^{\mathrm{tr}})\}+ p_h(C)\lambda_s\eta W\log_2(1+\gamma_0)P^{cov}_{s}(P_{\mathrm{s}}^{\mathrm{tr}}),
  \end{align}
where \textcolor{blue}{$P^{cov}_{s}(P_{\mathrm{s}}^{\mathrm{tr}})=\sum_{i=L,NL}P^{cov}_{s,i}(P_{\mathrm{s}}^{\mathrm{tr}})$,
	$P^{cov}_{bh}(P_{\mathrm{m}}^{\mathrm{tr}})=\sum_{i=L,NL}P^{cov}_{bh,i}(P_{\mathrm{m}}^{\mathrm{tr}})$  denotes the SINR coverage probability of wireless access link for user associated with SBS tier, wireless backhaul link for SBS associated with MBS, respectively.}  The Symbol $\min\{\}$ means that the minimum value of throughput in the wireless access link \textcolor{blue}{of SBS tier} and throughput in the wireless backhaul link.
\textcolor{blue}{The APT of MBS tier is given as follows.}
\begin{align}
\textcolor{blue}{\mathcal{R}_{m}(\eta)  = \lambda_m\eta W\log_2(1+\gamma_0)P^{cov}_{m}(P_{\mathrm{m}}^{\mathrm{tr}})}	
\end{align}
 \textcolor{blue}{where $P^{cov}_{m}(P_{\mathrm{m}}^{\mathrm{tr}})=\sum_{i=L,NL}P^{cov}_{m,i}(P_{\mathrm{m}}^{\mathrm{tr}})$ is the SINR coverage probability of wireless access link for user associated with MBS tier. Then, APT of mABHetNets can be obtained as follows.
 \begin{align} \label{APTHetNets}
 	&\mathcal{R}(\eta,C,P_{\mathrm{s}}^{\mathrm{tr}})  =\mathcal{R}_{s}(\eta,C,P_{\mathrm{s}}^{\mathrm{tr}}) +\mathcal{R}_{m}(\eta) 
\end{align}}
Note that, \textcolor{blue}{$C$ is the caching capacity of SBS} and $p_h = p_h(C) =\frac{\sum_{f=1}^{C}f^{-\gamma_p}}{\sum_{g=1}^{F}g^{-\gamma_p}}$ is the cache hit ratio in the SBS tier. $(1-p_h)$ reflects the probability   that the f\mbox{}iles  are not cached in SBS tier and need to be delivered through the backhaul link.  $P^{cov}_{s,\mathrm{L}}$ and $P^{cov}_{k,\mathrm{NL}},k \in \{s,m,bh\}$ is the SINR distribution of wireless LOS link and NLOS link in Proposition \ref{ProUserSINRCoverageProbability}.
\end{corollary}

\section{Problem Formulation and Solution}
In this section, we propose an  APT maximization problem with multivariate function and complicated integral component. Inspired by BCD method, we propose a two-step approach to solve it alternately and decompose the joint optimization problem into two sub-problems  to obtain the optimal spectrum partition coef\mbox{}f\mbox{}icient, power allocation as well as cache decision, respectively. Therefore, we can solve these two small-scale sub-problems iteratively.
\subsection{APT maximization problem}
%
\textcolor{blue}{Based on the above theoretical analysis with stochastic geometry, in order to further explore the potential performance superiority of the cache-enabled mABHetNets, we formulate a joint optimazation problem of cache, transmit power and spectrum partition to maximize the APT.}
\begin{subequations}
\begin{align}
  {\mathcal{P}:} \mathop{\max }\limits_{C,\eta,P_{\mathrm{s}}^{\mathrm{tr}}}\text{ }\min &\{(1-{{p}_{h}}\text{)}{{\lambda }_{m}}(1-\eta)W{{\log }_{2}}(1+{{\gamma }_{0}})P_{\textcolor{blue}{bh}}^{\operatorname{cov}}(P_{\mathrm{m}}^{\mathrm{tr}}),
   (1-{{p}_{h}}\text{)}{{\lambda }_{s}}\eta W{{\log }_{2}}(1+{{\gamma }_{0}})P_{s}^{\operatorname{cov}}(P_{\mathrm{s}}^{\mathrm{tr}})\}\nonumber \\
  &+{{p}_{h}}{{\lambda }_{s}}\eta W{{\log }_{2}}(1+{{\gamma}_{0}})P_{s}^{\operatorname{cov}}(P_{\mathrm{s}}^{\mathrm{tr}})\textcolor{blue}{+{{\lambda }_{m}}\eta W{{\log }_{2}}(1+{{\gamma}_{0}})P_{m}^{\operatorname{cov}}(P_{\mathrm{m}}^{\mathrm{tr}})} \label{Pa} \\
  s.t.\ &\rho_s P_{\mathrm{s}}^{\mathrm{tr}}+P_{\mathrm{s}}^{\mathrm{fc}}+\omega_{\mathrm{ca}}\textcolor{blue}{s}C \le P_{s}^{\max }, \label{Pb}\\ 
  &P_{s}^{tr}\ge 0, \label{Pc} \\
  &C\textcolor{blue}{\in}\{0,1,2,...,{C}_{\mathrm{max}}\}, \label{Pd} \\
  &\eta \in [0,1].  \label{Pe}
\end{align}
\end{subequations}
Constraint (\ref{Pb}) guarantees that the power consumed by the SBS will not exceed the maximum power constraint $P_{s}^{\max }$. Constraint (\ref{Pc}) ensures that the transmitted power is non-negative.
Constraint (\ref{Pd}) denotes that the number of the f\mbox{}iles cached at the SBS is a discrete variable form $0$ to ${C}_{max}$, where ${C}_{max}$ is the maximum caching capacity (file units) of the SBS.
Constraint (\ref{Pe}) determines that spectrum partition coef\mbox{}f\mbox{}icient is a continuous variable from $0$ to $1$.

From the above expression, we can easily get that the proposed  problem $P$ is a max-min and a mixed-integer nonlinear programming (MINLP) problem, which is non-linenar and non-convex. To simplify the problem and facilitate the solution, we introduce an auxiliary variable $Y$ to convert the problem $P$ into a more tractable form as $\mathcal{P}1$.
\begin{small}
\begin{subequations}
\begin{align}
  &{\mathcal{P}1:} \mathop{\max }\limits_{C,\eta,P_{s}^{tr} }  \quad Y  \label{P1a}\\
  s.t.\  &\text{(1-}{{p}_{h}}\text{)}{{\lambda }_{m}}(1-\eta )W{{\log }_{2}}(1+{{\gamma }_{0}})P_{\textcolor{blue}{bh}}^{\operatorname{cov}}(P_{\mathrm{m}}^{\mathrm{tr}})+{{p}_{h}}{{\lambda }_{s}}\eta W{{\log }_{2}}(1+{{\gamma }_{0}})P_{s}^{\operatorname{cov}}(P_{\mathrm{s}}^{\mathrm{tr}})\textcolor{blue}{+{{\lambda }_{m}}\eta W{{\log }_{2}}(1+{{\gamma}_{0}})P_{m}^{\operatorname{cov}}(P_{\mathrm{m}}^{\mathrm{tr}})}\ge Y,  \label{P1b}\\
  &\text{(1-}{{p}_{h}}\text{)}{{\lambda }_{s}}\eta W{{\log }_{2}}(1+{{\gamma }_{0}})P_{s}^{\operatorname{cov}}(P_{\mathrm{s}}^{\mathrm{tr}})+{{p}_{h}}{{\lambda }_{s}}\eta W{{\log }_{2}}(1+{{\gamma }_{0}})P_{s}^{\operatorname{cov}}(P_{\mathrm{s}}^{\mathrm{tr}})\textcolor{blue}{+{{\lambda }_{m}}\eta W{{\log }_{2}}(1+{{\gamma}_{0}})P_{m}^{\operatorname{cov}}(P_{\mathrm{m}}^{\mathrm{tr}})}\ge Y,  \label{P1c}\\
  &(\ref{Pb}),(\ref{Pc}),(\ref{Pd}),(\ref{Pe}). \notag
\end{align}
\end{subequations}
\end{small}
\begin{proposition}\label{proposition3}
   APT in cache-enabled millimeter wave HetNets will be maximized when the SBS consumes the maximum power, i.e.,satisfy
  \begin{align}
    \rho_s P_{s}^{\mathrm{tr}}+P_{\mathrm{s}}^{\mathrm{fc}}+\omega_{\mathrm{ca}}\textcolor{blue}{s}C = P_{s}^{\max }.
  \end{align}

  \begin{proof}
    When the power consumption satisfies $\rho_s P_{s}^{\mathrm{tr}}+P_{\mathrm{s}}^{\mathrm{fc}}+\omega_{\mathrm{ca}}\textcolor{blue}{s}C < P_{s}^{\max }$, the SBS can increase its caching power to $P_{s}^{\max}-\rho_s P_{s}^{\mathrm{tr}}-P_{\mathrm{s}}^{\mathrm{fc}}$ for a higher cache hit ratio , so intruducing the caching without reducing the transmit power can increase the APT. Based on proof by contradiction, we can get that the SBS comsumes the maximum power when  APT is maximized. Thus transmit power at the SBS can be denoted as $P_s^{tr}=P_s'- \omega_{ca}^{'s}C$
  \end{proof}
\end{proposition}

\begin{remark}
	From the above analysis, increasing the caching capacity leads to the spectrum transfer and increasing the APT. However, from constraint (\ref{Pb}), we can see that APT is not a monotonically increasing function of $C$. Transmit power and caching capacity are coupled variables. This reveals that increasing the caching capacity is adverse to improving the APT when  APT is limited by the transmit power.  
\end{remark}

\subsection{Sub-problem: Spectrum Partition Problem}
In this subsection, we will give the solution of the spectrum partition problem when the cache decision and tramsmit power is given. To simplify the notation, we stipulate that $A_{1}=\lambda_m W \operatorname{\log}_{2}(1+\gamma_{0}), A_{2}=\lambda_s W {\log}_{2}(1+\gamma_{0}), A_{3}=\lambda_m W \operatorname{\log}_{2}(1+\gamma_{0}) $. 
Then the original problem can be transformed to the spectrum partition problem as $\mathcal{P}2$.
\begin{subequations}
\begin{align}
	{\mathcal{P}2:}\mathop{max} \limits_{Y,\eta}  \quad &Y \label{P2} \\
	s.t.\  &A_1(1-p_h)(1-\eta)P_{bh}^{\operatorname{cov}}+A_{2} p_h P_{s}^{\operatorname{cov}} \eta\textcolor{blue}{+A_{3} \eta P_{m}^{\operatorname{cov}}} \ge Y,  \label{P2a}\\
	&A_{2} \eta P_{s}^{\operatorname{cov}}\textcolor{blue}{+A_{3} \eta P_{m}^{\operatorname{cov}}} \ge Y, \label{P2b}  \\
	&(\ref{Pb}),(\ref{Pc}),(\ref{Pd}),(\ref{Pe}).  \notag
\end{align}
\end{subequations}
of which $P_{k}^{\operatorname{cov}},k=\{s,  m, bh \}$  is an integral term containing the transmit power of the SBS or MBS. It's obvious that this sub-problem becomes a linear constraint programming problem of variable $Y$ and $\eta$, of which $Y$ is the corresponding APT and $\eta$ is the solution of spectrum partition coef\mbox{}f\mbox{}icient. We can solve this problem by interior point method to get the  $\eta^{*}$ and the  corresponding $Y^{*}$ directly. \textcolor{blue}{The spectrum partition problem is based on interior-point method to obtain the optimal spectrum partition coefficienct (one variable) with complexity $\mathcal{O}(1)$.}
\subsection{Sub-problem: Cache Decision and Power Allocation Problem}
In this subsection, we will optimize the cache decision $C$ and power allocation $P_{s}^{tr}$ problem based on the given spectrum partition coef\mbox{}f\mbox{}icient $\eta$, which is still a discrete, nonlinear and non-convex optimization problem. Based on an genetic algorithm, we propose a genetic-based cache decision and power allocation algorithm  which is elaborated in Algorithm \ref{alg:algorithm1} to give the approximate optimal solution in feasible time. Therefore, the cache decision and power allocation problem is reformulated  as $\mathcal{P}3$.
\begin{subequations}
\begin{align}
 {\mathcal{P}3:}\mathop{max} \limits_{Y,C,P_{s}^{tr}}  \quad &Y \label{P3a} \\
  s.t.\  &A_1(1-p_h)(1-\eta)P_{bh}^{\operatorname{cov}}+A_{2} p_h \eta P_{s}^{\operatorname{cov}}(P_s^{tr})\textcolor{blue}{+A_{3} \eta P_{m}^{\operatorname{cov}}} \ge Y,  \label{P3b}\\
  &A_{2} \eta P_{s}^{\operatorname{cov}}(P_s^{tr})\textcolor{blue}{+A_{3} \eta P_{m}^{\operatorname{cov}}} \ge Y,   \label{P3c}  \\
  &(\ref{Pb}),(\ref{Pc}),(\ref{Pd}),(\ref{Pe}).  \notag
\end{align}
\end{subequations}

In order to simplify the problem even further and from Proposition \ref{proposition3}, we stipulate that $f_{1}(C)=A_1(1-p_h(C))(1-\eta)P_{bh}^{\operatorname{cov}}+A_{2} p_h(C) \eta P_{s}^{\operatorname{cov}}(P_s'- \omega_{ca}^{'s}C)+A_{3} \eta P_{m}^{\operatorname{cov}}$,  $f_{2}(C)=A_{2} \eta P_{s}^{\operatorname{cov}}(P_s'- \omega_{ca}^{'s}C)+A_{3} \eta P_{m}^{\operatorname{cov}}$, then we optimize the minimum value of these two constraint functions according to an Genetic-based Cache Decision and Power Allocation (GCDPA) algorithm. \textcolor{blue}{This  algorithm is based on heuristic algorithm to make cache decision designed for SBS. The worst case for cache decision algorithm at each SBS is to traverse all feasible solutions with complexity $\mathcal{O}(C_\text{max})$. }
\begin{algorithm}[t]
	\caption{Genetic-based Cache Decision and Power Allocation Algorithm(GCDPA)}
	\label{alg:algorithm1}
	\KwIn{Initial $\mathcal{C}=\mathcal{C}_1$;  Spectrum partition ratio $\eta$ from f\mbox{}irst sub-problem; \\
		Generation size: $\mathcal{S}^{g}$; Chromosome size: $\mathcal{S}^{c}$; Population size: $\mathcal{S}^{p}$; Crossover rate: $\mathcal{\rho}_{c}$; Mutation rate: $\mathcal{\rho}_{m}$; Fitness functions $f(\mathcal{C})=min\{f_{1}(\mathcal{C}),f_{2}(\mathcal{C})\}$; }
	\KwOut{Best individual $\mathcal{C}^{*}$; Best f\mbox{}itness value $f^{*}(\mathcal{C})$;}
	Optimize f\mbox{}itness function;
	\BlankLine

		$\mathcal{C}^{1}$=Initialization($\mathcal{S}^{c}$,$\mathcal{S}^{p}$);	
		
		\For{$k=1;k<\mathcal{S}^{g};k++$}{
			$\mathcal{C}^{*}$, $f^{*}(\mathcal{C})$ = Fieness Function($\mathcal{C}^{k}$);\\
			${\mathcal{C}^{k}}$ = Selection($\mathcal{C}^{k},f^{*}(\mathcal{C})$);\\
			${\mathcal{C}^{k}}$ = Crossover($\mathcal{C}^{k},\mathcal{\rho}_{c}$);\\
			${\mathcal{C}^{k}}$ = Mutation($\mathcal{C}^{k},\mathcal{\rho}_{m}$);\\
			%
			%
			
			Return $\mathcal{C}^{*}, f^{*}(\mathcal{C}) $;		
	}

	Output $\mathcal{C}^{*},f^{*}(\mathcal{C})$;\\
		
\end{algorithm}

\subsection{Joint optimization solution}
In view of the fact that this joint optimization problem is a non-convex and multivariate problem, we propose an alternating optimization algorithm to optimize spectrum partition and cache decision, which is described   in Algorithm \ref{alg:algorithm2}. The idea of iterative optimization is inspired by the BCD method proposed in \cite{BCDproof}. BCD is a method for optimizing multivariate function even if it is neither necessarily strictly convex nor differentiable \cite{BCDdiff}. In each iteration, it takes the extremum  along the direction of multiple coordinate axes (variables), and one of the coordinates  or the coordinate blocks is f\mbox{}ixed to optimize another coordinate or coordinate block, then the new result is immediately substituted into the next iteration. Particularly, the block coordinate iteration method will converge in a f\mbox{}inite number of times as long as the two variables have no strong correlation. \textcolor{blue}{The JCSPA algortithm is based on BCD method to optimize subproblem of spectrum partition and subproblem of cache decision and power allocation alternately, so the worst case is to iterate to the maximum iteration times $Iter_{\text{max}}$. The complexity of entire JCSPA algorithm  is upper bound by $\mathcal{O}(C_{\text{max}}*Iter_{\text{max}})$.} We further illustrate  the optimality of Algorithm \ref{alg:algorithm2} in the following theoreom.
\begin{algorithm}[t]
	\caption{Joint Cache Decision, Spectrum Partition and Power Allocation Altertive Optimization Algorithm(JCSPA) }
	\label{alg:algorithm2}
	\KwIn{Parameter of simulation scenario; Maximum number of iterations $Itermax$; $\epsilon$}
	\KwOut{Optimal cache decision $\mathcal{C}^{*}$ and spectrum partition $\mathcal{\eta}^{*}$.}
	\BlankLine
	\textbf{Initial procedure}:\\
	Set $t=0$;\\
	Strating point $\mathcal{C}^{(0)}$;\\
	$APT_{1}^{(0)}=0$, $APT_{2}^{(0)}=0$;

	\textbf{end Initial procedure};

	\Repeat{$t = Itermax$ or $||APT_{2}^{(t)}-APT_{1}^{(t)}||\le\epsilon$ }{
		Set $t = t+1$;\\
		Update $\eta^{(t)}$ for fixed $\mathcal{C}^{(t-1)}$ according to (\ref{P2}) and send it to next sub-problem,
		then calculate the corresponding $APT_{1}^{(t)}$ according to (\ref{Pa}); \\	
		Solve sub-problem of cache decision for given $\eta^{(t)}$ based on Algorithm \ref{alg:algorithm1}  to get $\mathcal{C}^{(t)}$, 
		then calculate the corresponding $APT_{2}^{(t)}$ according to (\ref{Pa}), and send $\mathcal{C}^{(t)}$ to the $(t+1)$th iteration; 
		
	}
	
	%
	%
	
	Compute $P_s^{tr*}$ according to Proposition \ref{proposition3};
	
	\textbf{Return} $\mathcal{C}^{*}$,$\mathcal{\eta}^{*}$
\end{algorithm}

\begin{theorem}\label{BCDconvergence}
	Algorithm \ref{alg:algorithm2} based on BCD method will converge in f\mbox{}inite number of iterations and  the convergent output $\eta^{*},C^{*}$ is the \textcolor{blue}{approach} to the optimal solution of the original problem P.
\end{theorem}
	This theorem can be proved that these two sub-problems have a unique optimal solution.\textcolor{blue}{ Then based on basic cyclic rule, the consequences from BCD method is  defined and bounded }and the constraint set is a Cartrsian product of closed convex sets. \textcolor{blue}{In proposed JCSPA algorithm, after Proposition 3, the spectrum partition coefficient $\eta$  and caching decision $C$ are optimized in turn while keeping the other fixed, with the goal of maximizing APT. After spectrum partition solution for fixed $C_t$, it will hold that 
		\begin{align} \label{covergenceA}
			APT(\eta_t, C_t) \le APT(\eta_{t+1}, C_t)
		\end{align}  
		Then after cache decision in Algorithm 1 for given $\eta_{t+1}$, we can get
		\begin{align} \label{covergenceB}
			APT(\eta_{t+1}, C_t) \le APT(\eta_{t+1}, C_{t+1})
		\end{align}
		Combining (\ref{covergenceA}) with (\ref{covergenceB}), it holds that 
		\begin{align} \label{covergence}
			APT(\eta_{t}, C_t) \le APT(\eta_{t+1}, C_{t+1})
		\end{align}
		This illustrates that APT value will not decrease in each iteration. The APT upper bound of cache-enabled mABHetNets is  limited to a finite value, so the proposed  JCSPA algorithm will  converge.}
\section{Performance Evaluation}
In this section, we use numerical simulation results to validate and evaluate of APT of the cache-enabled mABHetNet. Particularly, we compare APT under different scenarios.

\subsection{Parameter Setting}
%
Unless otherwise specif\mbox{}ied, the parameters of the simulation scenario are set as follows.
The BSs and users are distributed in a square area $1000\text{m} \times 1000\text{m}$.
Note that  the density of the users is assumed to be suf\mbox{}f\mbox{}iciently larger than that of the BS, we assume that the BS is active in the downlink. In other words, each SBS is in a backhaul association cell and has at least one associated user in its coverage. 
Some default simulation conf\mbox{}igurations are listed in Table \ref{parameters}, based on 3GPP specif\mbox{}ication and literatures\cite{Partition3,BSPowerModel,3GPP:Spec,PathLoss,FixPower,SBSPowerScaling}. All the above settings can be changed according to different scenarios.
\renewcommand\arraystretch{0.7}
\begin{table}[h]
	\footnotesize
\centering
\caption{Simulation Parameters}
\label{parameters}
 \begin{tabular}{|l|l|l|}
 \hline
 \textbf{Parameters}         &  \textbf{Physics meaning}            &  \textbf{Values}\\ 
 \hline
  $W$                                &  Total mmWave  bandwidth          &400 MHz\\ \hline
  $\lambda_s$,   $\lambda_m$  &  The density of SBS and MBS          &$10^{-4}$ BSs/m$^2$, ~$4 \times10^{-5}$ BSs/m$^2$ \\ \hline
  $A_\mathrm{L},\alpha_\mathrm{L}$ & Pathloss parameters of LOS &$10^{-10},~2$\\ \hline
  {\color{blue}$A_{\mathrm{NL}},\alpha_{\mathrm{NL}}$} & {\color{blue}Pathloss parameters of NLOS}  &{\color{blue}$10^{-14},~4$}\\ \hline
  $N_0$                          &  Noise power                &-90 dBm\\ \hline 
  $P_s^\mathrm{max}$, $P_m^\mathrm{max}$  &Maximum power of SBS and MBS           &38.2 dBm, ~60 dBm\\ \hline
  $\gamma_0$               &  SINR threshold            &10 dB \\ \hline
  $P_{s}^\mathrm{fc}$, $P_{m}^\mathrm{fc}$ &  Fixed circuit power at SBS and MBS       &20 dBm, 46 dBm\\ \hline
   $\rho_s$, $\rho_m$   &  Power amplif\mbox{}ier and cooling coeff\mbox{}icient    &1,~1.5\\ \hline
  $B_s$, $B_m$            &  Association biases of SBS and MBS         &10,~5\\ \hline
  $\beta$                         & Blockage density                                          &$2\times10^{-3}$\\ \hline
  {\color{blue}$\theta$} & {\color{blue}Main lobe beamwidth}      &{\color{blue}$30^{\circ}$}\\ \hline
  {\color{blue}$M$, $m$} & {\color{blue}Main lobe antenna gain, slide lobe antenna gain}  &{\color{blue}10dB, -10dB}\\ \hline
  $\omega_{ca}$           & Caching power coeff\mbox{}icient           & $6.25\times10^{-12}$ W/bit\\ \hline
  $s$                               &  Size of each f\mbox{}ile                            & 100 MB\\ \hline
  ${F}$                          &The f\mbox{}ile units in the f\mbox{}ile library &1000\\ \hline
  ${C}_{\mathrm{max}}$      & Maximum caching capacity(file units) of SBS &800\\ \hline
  $\gamma_p$               & Zipf parameter                                              &1.0\\ 
 \hline
 \end{tabular}
\end{table}

\subsection{Covergence of JCSPA}

\begin{figure}[H]
	\centering
	\includegraphics[width=3in]{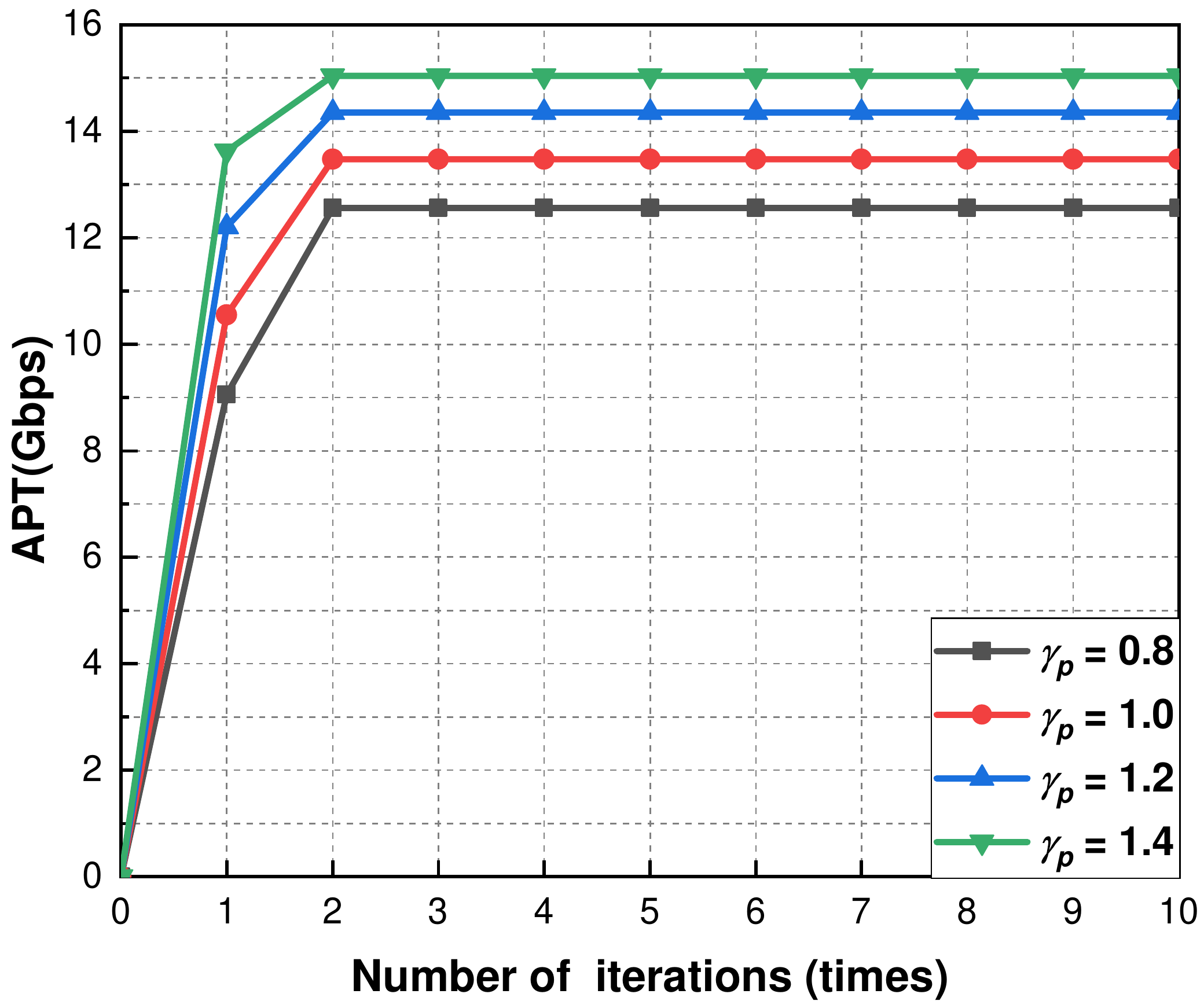}
	\captionsetup{name={Fig.},labelsep=period,singlelinecheck=off,font={small}}
	\caption{ Convergence performance of the proposed JCSPA algorithm in APT maximization problem.}
	\label{APT_Iteration1}
\end{figure}
In this subsection, simulations are carried out to validate the convergence performance and maximum APT of the proposed JCSPA algorithm  for different skewnesses of caching file popularity(from 0.8 to 1.4) as Fig. \ref{APT_Iteration1}. As expected, we can see that the joint optimization problem converges to maximum APT after approximately no more than five iterations. It can be concluded that our proposed algorithm is fast-convergent. Besides, compared to the different Zipf parameters, a bigher skewness will lead to a higher APT.

\subsection{Performance of APT under different caching capacity and spectrum partitions}
\begin{figure}[H]
	\centering
	\subfigure[]{\includegraphics[width=2.1in]{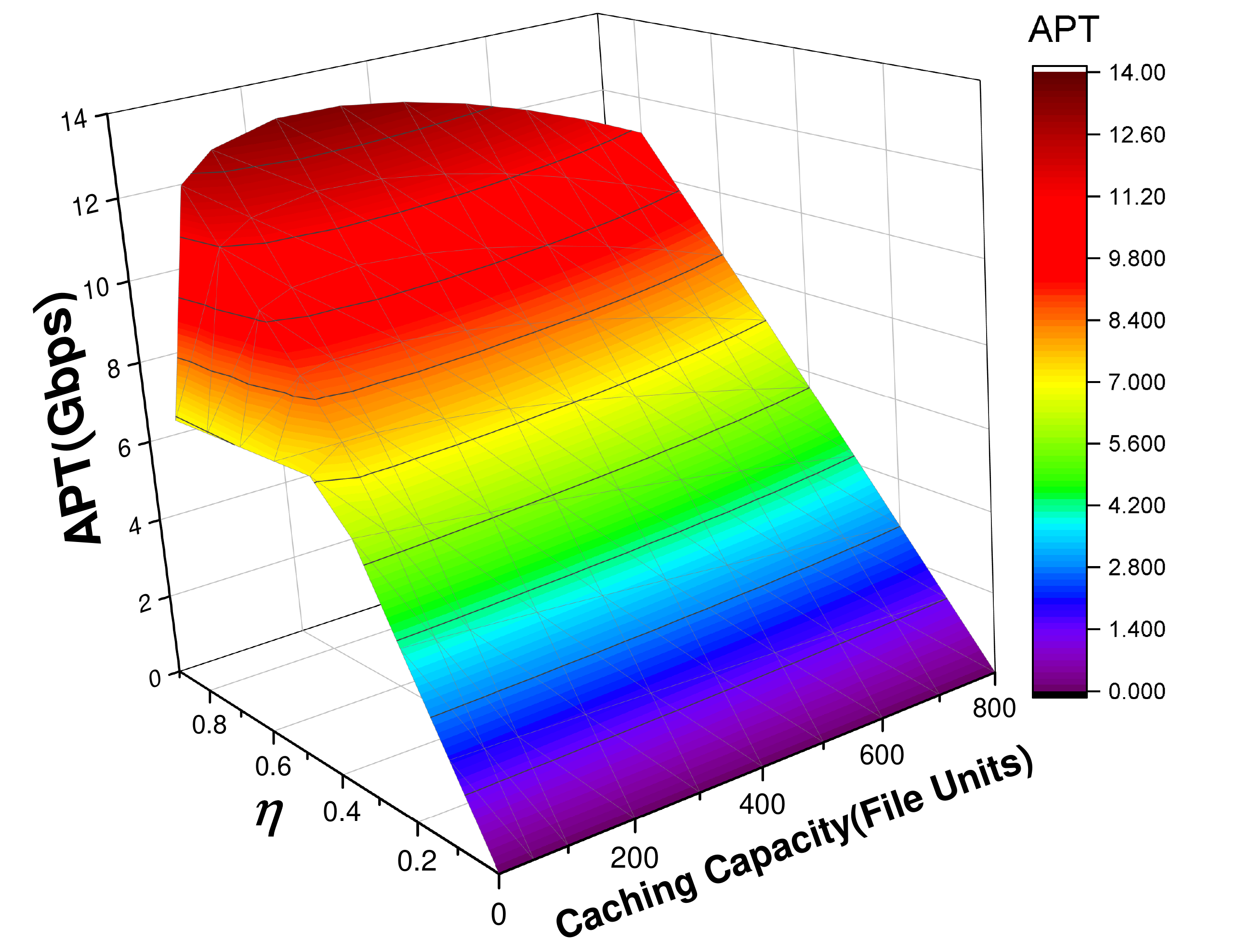}\label{APTetacache}}  
	\subfigure[]{\includegraphics[width=2.1in]{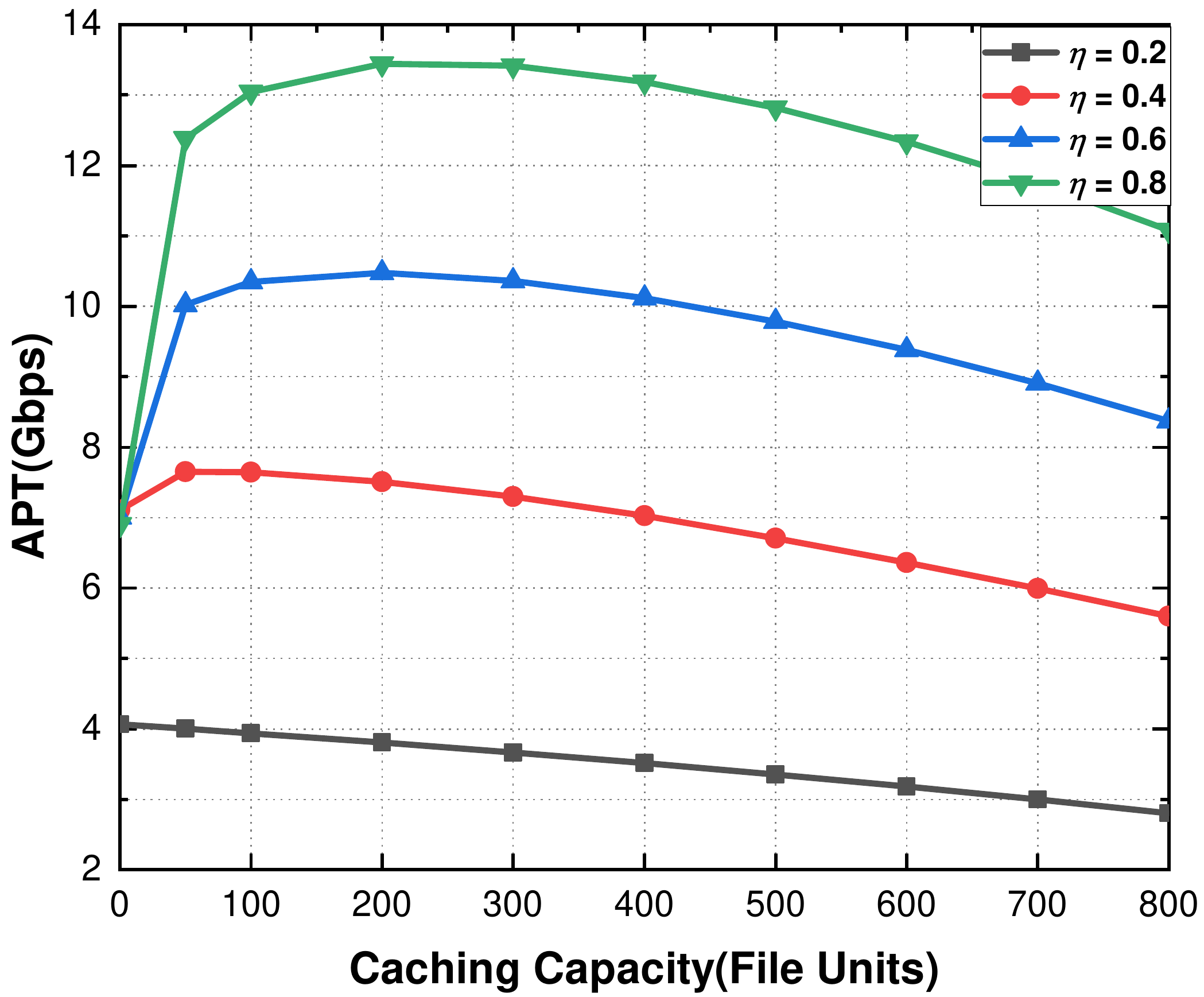}\label{APTcache}} 
	\subfigure[]{\includegraphics[width=2.1in]{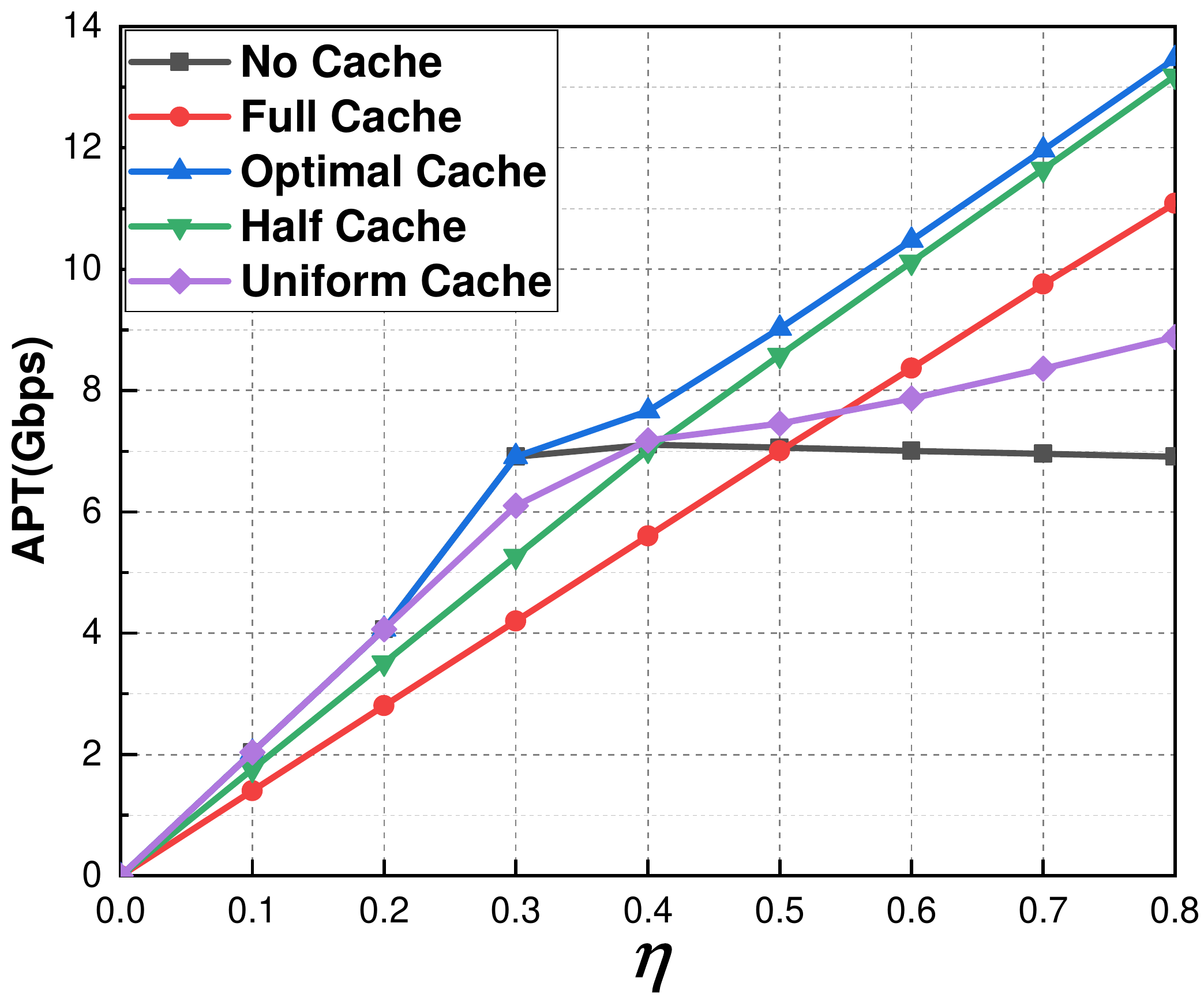}\label{APTeta}}      
	\captionsetup{name={Fig.},labelsep=period,singlelinecheck=off,font={small}}
	\caption{ APT under different caching capacity and spectrum partitions, 
	(a) APT v.s. caching capacity \& spectrum partition ratio, (b) APT v.s. caching capacity, (c) APT v.s. spectrum partition ratio.}
	\label{APT_cache_eta}
\end{figure}
In Fig. \ref{APTetacache}, the performance of APT for both the  caching capacity and spectrum partition is shown. From this three-dimensional APT, both caching f\mbox{}ile capacity and spectrum partition ratio have a prominent impact on APT. APT is not monotonically increasing  or decreasing with respect to these two variables. 

In order to further observe their quantitative relationship, Fig. \ref{APTcache} shows the APT with respect to different caching capacity under four spectrum partition coefficients. For lower spectrum partition ratio,  APT will decrease with increasing caching capacity. This is because the smaller spectrum can only satisfy the partial data transmission  of the access link, and the power consumption  by caching reduces the transmit power of the SBS, which further  reduces the APT. While for higher spectrum partition ratio, APT will f\mbox{}irst increase with increasing caching capacity and then decrease. 
The reason is that as more caching f\mbox{}iles can improve the cache hit ratio of files at SBS and less files' delivery occupy backhaul resources, more files can be obtained through access link directly without backhaul link. However,  APT will decrease when the caching capacity  is more than about 200 file units. 
This is mainly because  maximum power of SBS is limited and too much caching capacity consumes too much power, therefore the transmit power of SBS will be reduced, which limits the performance of APT.  This gives us an insight that increasing the caching capacity does not necessarily increase the APT when the maximum power of SBS is limited.     

 Fig. \ref{APTeta} shows the APT of different spectrum partitons under four given caching decesions. In the case of no cache, 
APT will increase and then decrease with increasing spectrum partition. While for higher spectrum partition ratio, the introduction of cache at SBS increases the APT signif\mbox{}icantly. This is because caching f\mbox{}iles at SBS can transfer the spectrum from the backhaul link to the access link, which potentially improves the APT.

\subsection{Comparison with baseline algorithms}
The proposed JCSPA algorithm is compared with  three basic algorithms in this subsection. 
\begin{itemize}
\item \textbf{No cache and DSA:}
This algorithm is used to optimize the spectrum partition while the SBS does not cache any f\mbox{}iles  in the network \cite{Partition3}. This means that all files are cache-miss and the f\mbox{}ile delivery needs to go through  both the access and backhaul link.
\item \textbf{Optimal cache and FSA:}
This algorithm aims to optimize the cache decision with given spectrum that backhaul and access link accounts for the half of the total spectrum \cite{OCF}. The SBS needs to decide how many most popular f\mbox{}iles in f\mbox{}ile library to cache.
\item \textbf{Full cache and DSA:} In this algorithm, the SBS uses all the caching capacity  to store the $C_{\mathrm{max}}$ most popular f\mbox{}ile units. Besides, the spectrum partition is based on DSA program to improve APT as much as possible.
\item \textcolor{blue}{
  \textbf{Uniform cache and DSA}: This baseline algorithm is based on caching files uniformly \cite{2u} where each file is cached with identical probability and spectrum partition scheme is dynamic spectrum partition.
}
\end{itemize}

\begin{figure}[H]
	\centering
	\subfigure[]{\includegraphics[width=2in]{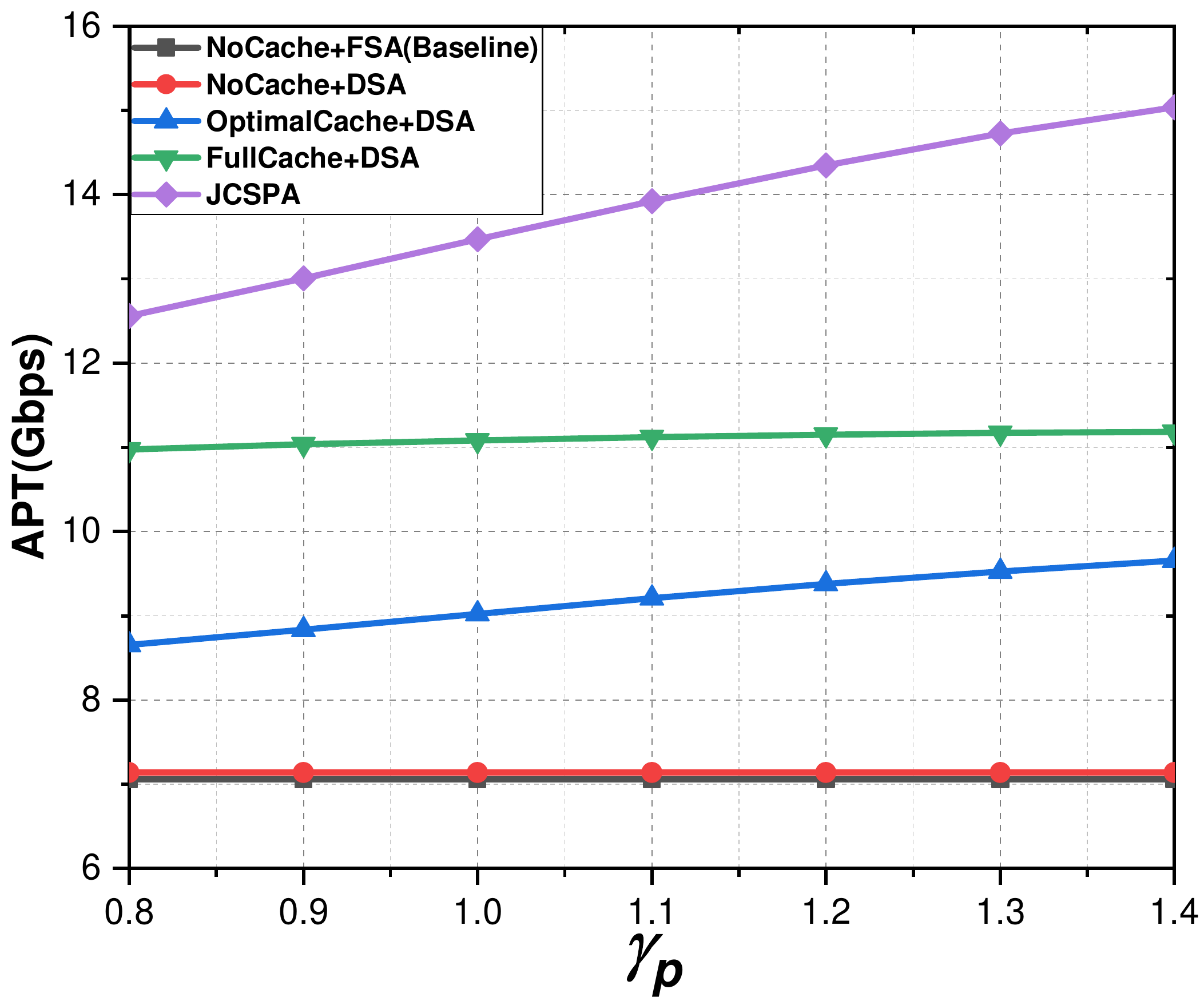}\label{APTgm}}  
	\subfigure[]{\includegraphics[width=2in]{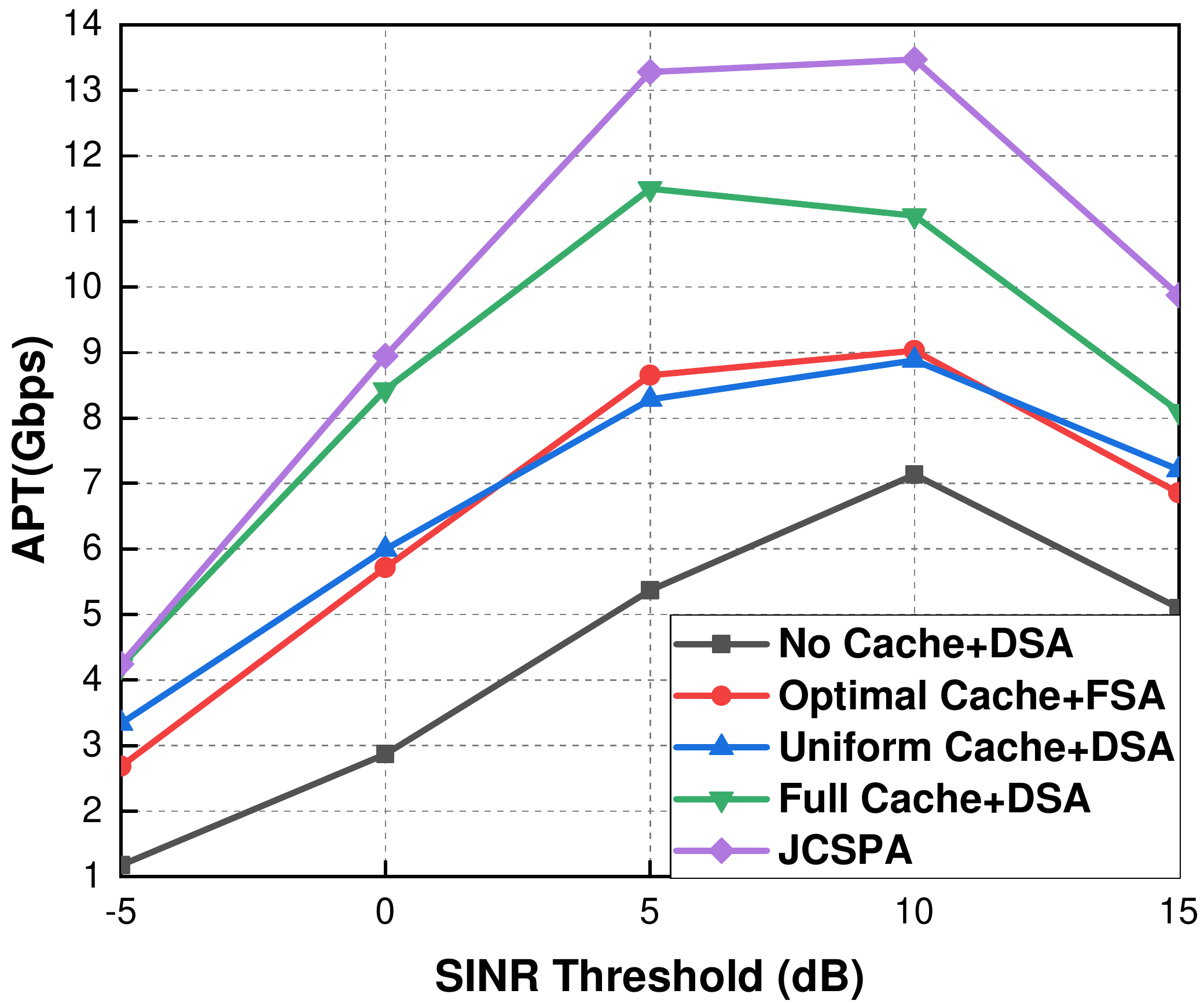}\label{APTsinr}}  
	\subfigure[]{\includegraphics[width=2in]{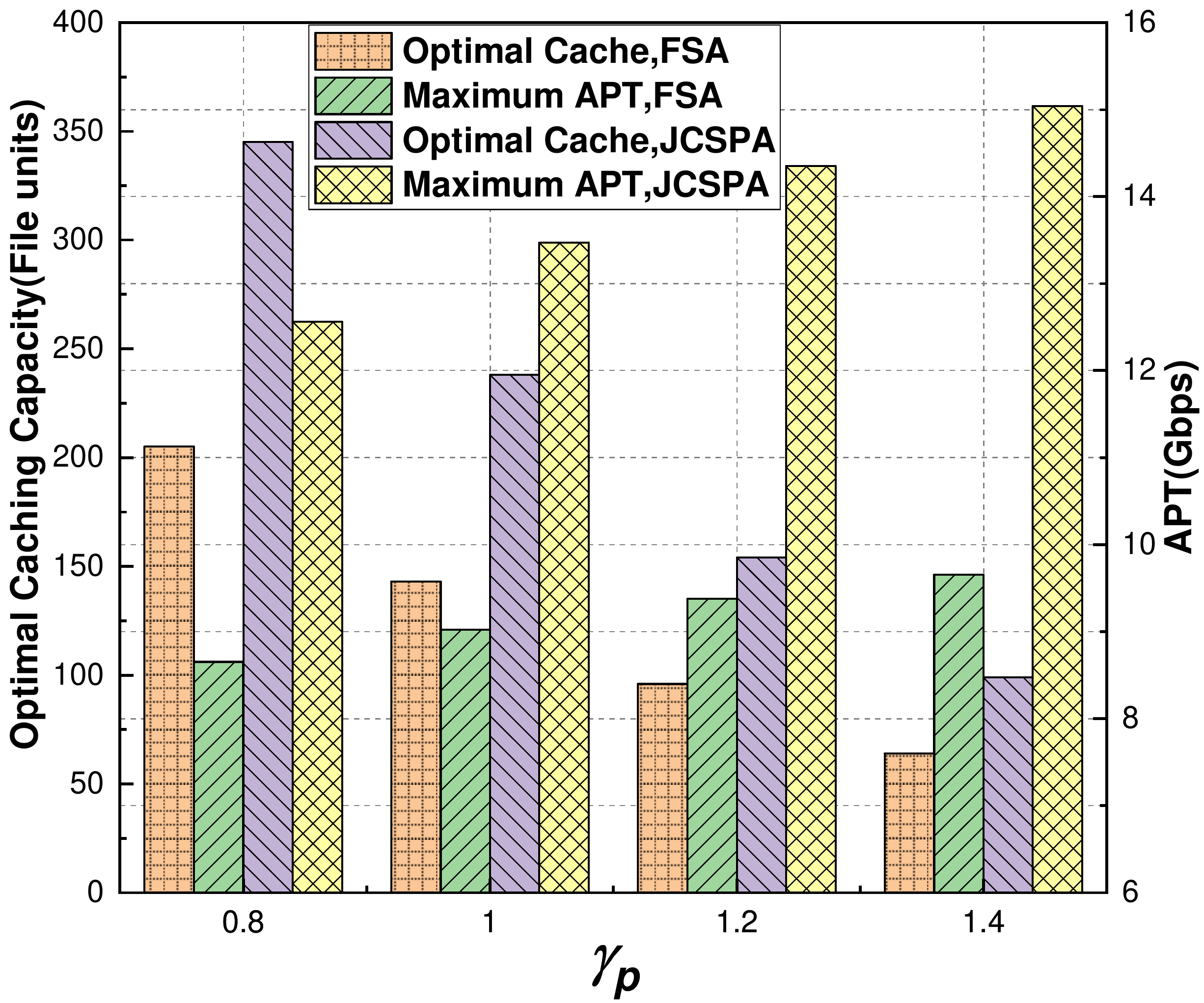}\label{APTcg}}    
	\captionsetup{name={Fig.},labelsep=period,singlelinecheck=off,font={small}}
	\caption{Comparison of  APT under different caching decision and spectrum partition algorithms,
		(a) APT v.s. skewness of f\mbox{}ile popularity, (b) APT v.s. SINR threshold(dB),  (c) Optimal cache \& APT v.s. skewness of f\mbox{}ile popularity. }
	\label{APTduibi}
\end{figure}
In Fig. \ref{APTgm}, we begin by evaluating APT in terms of different skewness of f\mbox{}ile popularity under proposed JCSPA and other three benchmark algorithms. In the cases of no cache, APT will not change with $\gamma_p$ but the DSA exceeds the FSA algorithm. Besides, APT of full cache combined with DSA is better than that of optimal cache with FSA. Such result reveals the signif\mbox{}icance of the dynamic spectrum partition. In contrast, our proposed JCSPA algorithm outperforms other baseline algorithms. This illustrates the effectiveness of proposed joint optimization of cache decision and spectrum partiton  JCSPA algorithm. 

Fig. \ref{APTsinr} shows the APT with different SINR thresholds under different algorithms.  Note that APT increases as SINR threshold increases when the SINR threshold is low, this is because increasing the SINR threshold $\gamma_0$ will increase the  transmission rate dominates. While for larger SINR threshold,  it improves the demodulation threshold of the received signal and reduces the coverage probability , so that APT will decrease with increasing the SINR threshold. Besides, the performance of JCSPA algorithm also signif\mbox{}icantly outperforms other baseline algoritms.


From Fig. \ref{APTcg} we can see that the optimal caching capacity is  reduced with the increasing the skewness while the maximum APT will increase. Comparing the two cases where  $\gamma_p=1.4$ and $\gamma_p=0.8$, APT under JCSPA algorithm is increased by nearly 40\% while the optimal caching capacity is reduced by more than 70\%.   In other words, only a smaller caching capacity is needed to reach the maximum APT for a higher skewness. The reason is that  higher skewness helps to improve cache hit ratio, so only fewer caching f\mbox{}ile units are needed to achieve higher APT. 

\subsection{Impact of other cache parameters on APT}
\begin{figure}[htbp]
	\centering
	\subfigure[]{\includegraphics[width=2.5 in]{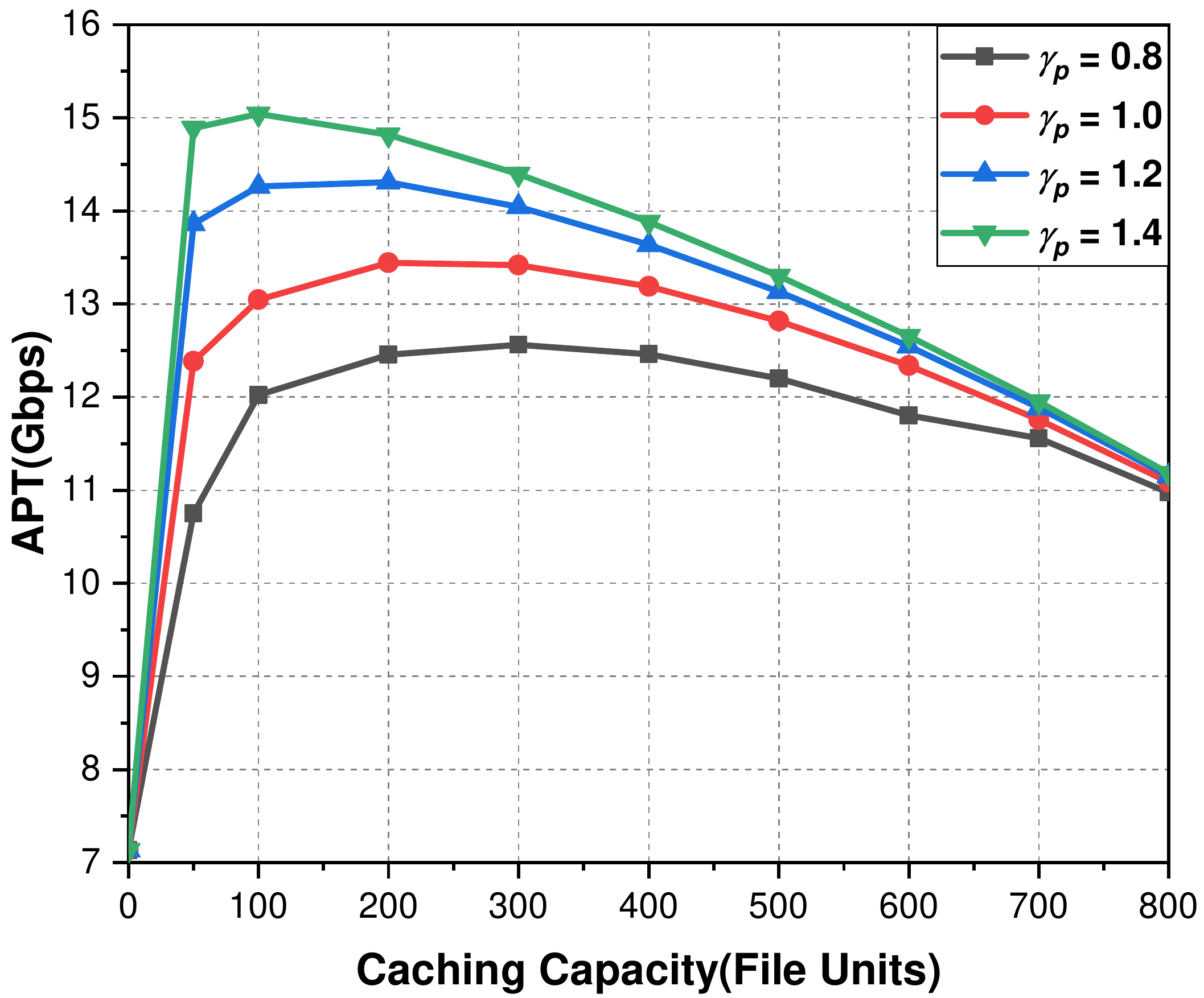}\label{APTcs}}  
	\subfigure[]{\includegraphics[width=2.5 in]{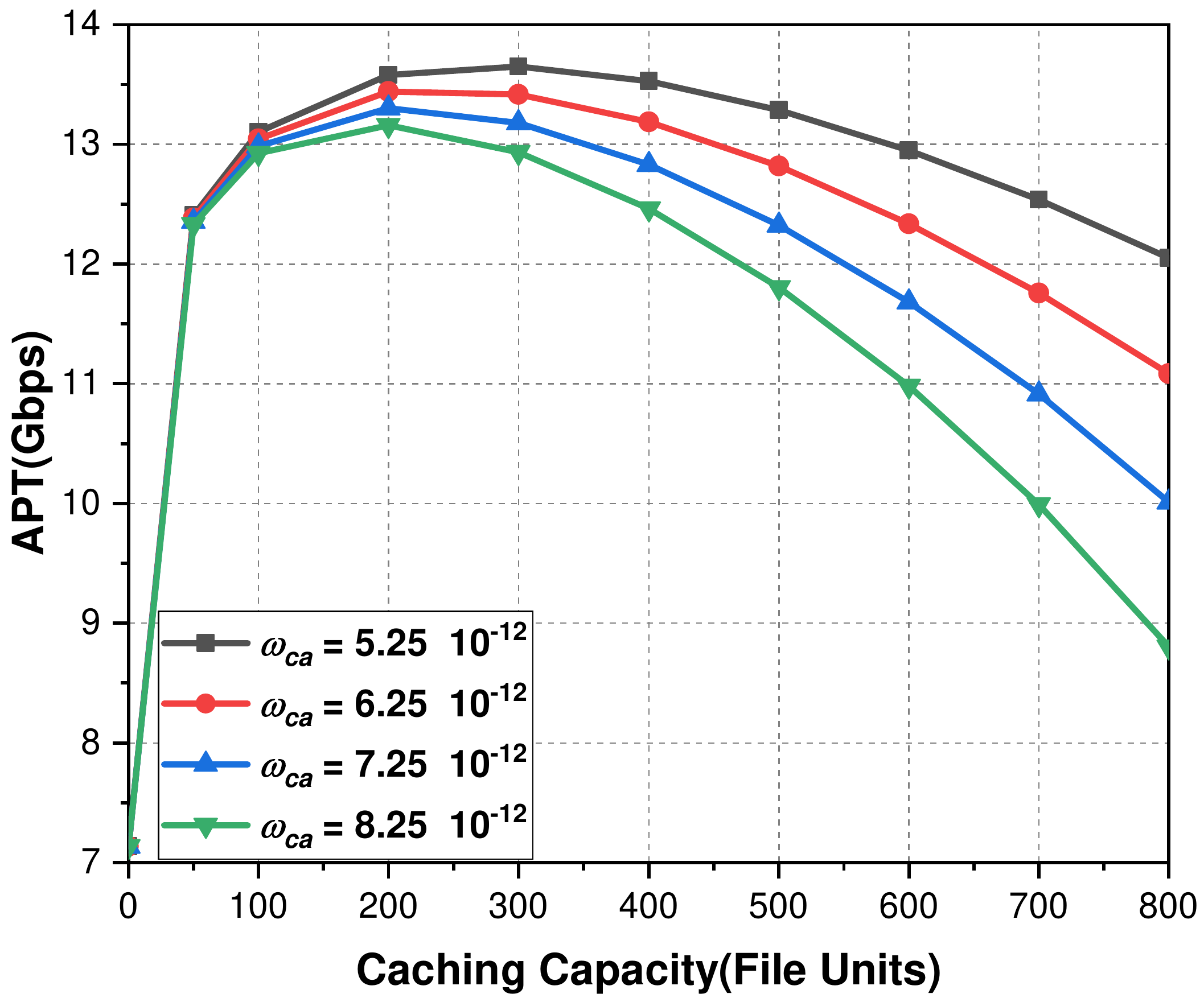}\label{APTcw}}        
	\captionsetup{name={Fig.},labelsep=period,singlelinecheck=off,font={small}} 
	\caption{Impact of cache parameters on APT,
		(a) APT vs. Caching capacity with different skewness $\gamma_p$, (b) APT vs. Caching capacity with different caching power coef\mbox{}f\mbox{}icient $\omega_{ca}$.}
	\label{EECache}
\end{figure}
In order to obtain the impact of cache-related parameters on APT, we further give the APT of different skewness of f\mbox{}ile popularity and caching power coef\mbox{}f\mbox{}icients. The skewness $\gamma_p$ and caching power coef\mbox{}f\mbox{}icient reflect the main characteristics of the caching. Fig. \ref{APTcs} shows the APT under dif\mbox{}ferent caching capacity for four skewness of caching f\mbox{}iles. It can be observed that APT f\mbox{}irst increases with caching capacity and then decreases, as well as higher skewness can lead to higher APT. This is because that  higher skewness increases the cache hit ratio $p_h$ then increases the APT.
Fig. \ref{APTcw} shows the APT under different caching capacity for four caching power coef\mbox{}f\mbox{}icients  from $\text{5.25}\times \text{10}^{\text{-12}}$ to 
$\text{8.25}\times \text{10}^{\text{-12}}$. When the caching capacity is small, APT is almost equal under different cache power consumption coefficients, but with the increasing of caching capacity, for the case of high cache power consumption coef\mbox{}f\mbox{}icient, APT will decrease rapidly as the cache increases, while APT with lower cache power consumption coef\mbox{}f\mbox{}icient decreases slowly. This is because the maximum power of SBS is limited and higher $\omega_{ca}$ will lead to more caching power consumption for the same number of caching f\mbox{}iles, thereby reducing the transmit power and decreasing the APT. This prompts us to use storage technologies with lower caching power coef\mbox{}f\mbox{}icient to improve network performance.

\section{Conclusion}
In this paper, to cope with the ``\emph{spectrum occupancy}'' problem, we make a tractable  cache-enabled mABHetNet  by stochastic geometry tool and investigate the impact of spectrum partition between the access link and the backhaul link as well as cache allocation on network performance. \textcolor{blue}{Considering the effect of blockage and association probability}, we analyze the SINR distribution and derive the expression of APT with respect to caching capacity, spectrum partition and SINR threshold.  Based on the analytical work, we optimize the cache decision and spectrum partition problem jointly to obtain the maximum APT by the proposed algorithm.  Via numerical evaluations, we f\mbox{}irst evaluate the convergence of our proposed algorithm. Then we find that there exist the  optimal cahing capacity and spectrum partition to maximize APT, which verifies the effectiveness of our proposed algorithm.  Besides, we explore the impact of some key cache-related factors on APT. The results show  \textcolor{blue}{up to 90\% APT} gain of appropriate caching capacity and spectrum partition \textcolor{blue}{compared with traditional mABHetNets.} 
\section{Appendix}

\subsection{The proof of Lemma \ref{LemmaUserClosestSBSProbabilty}}\label{Appe1:NearDistance}
We first consider the event that the    distance between the typical user  and  the nearest LOS SBS   (LOS based file transmission between the typical user and the SBS)   is $r$.
In fact,   the   event that is the joint of following two events:
The first event is  the nearest SBS of the typical user is located at distance $r$ (Event 1) and the second event is the transmission path between the typical user and the serving SBS is an LOS path (Event 2).
According to \cite{ReferUserClosestBSProbabilty}, the PDF of Event 1 with regard to $r$ is given by $\exp \left(-\pi r^{2} \lambda_{s}\right) \times 2 \pi r \lambda_{s}$. The probability of Event 2 over distance $r$  is $\mathcal{P}_L(r)$, so that we can get the PDF of the joint Event 1   and Event 2 as
\begin{align}
f_{R_{s}}^{\mathrm{k}}(r)= \mathcal{P}_{k}(r)\times 2 \pi r \lambda_{s} \label{RsLOS} \times   \exp \left(-\pi r^{2} \lambda_{s}\right),
\end{align}
where $k=\{L\ or\ NL\}$ denotes the LOS or NLOS transmission link.

Similarly, the PDF of the distance $r$ (between the SBS and  the associated LOS MBS or NLOS MBS )  is
\begin{align}
f_{R_{m}}^{\mathrm{k}}(r)&=\mathcal{P}_{\mathrm{k}}(r)   \times 2 \pi r \lambda_{m} \label{RmLOS} \times \exp \left(-\pi r^{2} \lambda_{m}\right). 
\end{align}

\subsection{The proof of Lemma \ref{LemmaUserAssociatedSBS}}\label{Appen2:AssoPro}
The typical user may be associated with the SBS tier by either LOS channel or NLOS channel.
We derive the probability of the first event that the user is associated with the SBS by the wireless LoS  link. Such event has other cases: the interference from the NLOS SBS and the interference from the  LoS. Therefore, the probability that the user obtain the desired LOS signal from the SBS is $F_{s}^{L}(r)=p_{ln}^{ss}(r)\textcolor{blue}{p_{ll}^{sm}(r)p_{ln}^{sm}(r)}f_{R_s}^{\mathrm{L}}(r)$, 
where
\begin{enumerate}
    \item   \textcolor{blue}{When the recieved power from LOS SBS is higher than that from NLOS SBS }, the user is associated with the LOS SBS and the interference is from NLOS SBS.
            \begin{align}\label{LSBSNLSBS}
             p_{ln}^{ss}(r)
             &=\mathbb{P}[P_{s}^{tr}  B_s\textcolor{blue}{G_s}h_{s}  A_\mathrm{L} r^{-\alpha_L} \geq P_{s}^{tr}  B_s\textcolor{blue}{G_s}h_{s} A_{\mathrm{NL}} r_{s}^{-\alpha_{\mathrm{NL}}}]\\
             &=\mathbb{P}\left[ r_{s}\geq\left(\frac{A_\mathrm{L}}{A_\mathrm{NL}}\right)^{\frac{-1}{\alpha_{NL}}} r^{\frac{\alpha_{L}}{\alpha_{NL}}} \right]
             =\int_{0}^{\left(\frac{A_\mathrm{L}}{A_\mathrm{NL}}\right)^{\frac{-1}{\alpha_{NL}}}r^{\frac{\alpha_{L}}{\alpha_{NL}}}}f_{R_s}(r)dr
             =e^{-\lambda_s\pi\left[ \left(\frac{A_\mathrm{L}}{A_\mathrm{NL}}\right)^{\frac{-1}{\alpha_{NL}}} r^{\frac{\alpha_{L}}{\alpha_{NL}}}\right]^{2}},
            \end{align}
            where the last step is based on the CDF of PPP in \cite{ReferUserClosestBSProbabilty}.
    \textcolor{blue}{\item   Similarly, the user is associated with the LOS SBS and the interference is from LOS MBS.
    	\begin{align}\label{LSBSLMBS}
    		p_{ll}^{sm}(r)
    		&=\mathbb{P}[P_{s}^{tr}B_s {G_s} h_{s}  A_\mathrm{L} r^{-\alpha_L} \geq P_{m}^{tr} B_{m}G_m h_m A_{\mathrm{L}} r_{m}^{-\alpha_L}]\\
    		&=\mathbb{P}\left[ r_{m}\geq\left(\frac{P_s^{tr}B_{s} {G_s}h_s}{P_m^{tr}B_{m}{G_m} h_m}\right)^{\frac{-1}{\alpha_{L}}} r \right]
    		=e^{-\lambda_m\pi\left[\left(\frac{P_s^{tr} B_{s}{G_s}h_s}{P_m^{tr} B_{m}{G_m} h_m}\right)^{\frac{-1}{\alpha_{L}}} r\right]^{2}},
    	\end{align}
    	\item   The user is associated with the LOS SBS and the interference is from NLOS MBS.
    	\begin{align}\label{LSBSLMBS}
    		p_{ln}^{sm}(r)
    		&=\mathbb{P}[P_{s}^{tr}B_{s}G_s h_{s}  A_\mathrm{L} r^{-\alpha_L} \geq P_{m}^{tr} B_{m}G_m h_{m} A_{\mathrm{NL}} r_{m}^{-\alpha_{\mathrm{NL}}}]\\
    		&=\mathbb{P}\left[ r_{m}\geq\left(\frac{P_s^{tr} B_{s}{G_s}h_s A_\mathrm{L}}{P_m^{tr} B_{m}{G_m}h_m A_\mathrm{NL}}\right)^{\frac{-1}{\alpha_{\mathrm{NL}}}} r^{\frac{\alpha_{\mathrm{L}}}{\alpha_{\mathrm{NL}}}} \right]
    		=e^{-\lambda_m\pi\left[\left(\frac{P_s^{tr} B_{s}{G_s}h_s A_\mathrm{L}}{P_m^{tr} B_{m}{G_m}h_m A_\mathrm{NL}}\right)^{\frac{-1}{\alpha_{\mathrm{NL}}}} r^\frac{\alpha_{\mathrm{L}}}{\alpha_{\mathrm{NL}}}\right]^{2}}.
    \end{align}}   
\end{enumerate}
Then,  the probability that the user obtain the desired NLOS signal from the SBS is
\begin{align*}
	F_{s}^{NL}(r)=p_{nl}^{ss}(r)\textcolor{blue}{p_{nl}^{sm}(r)p_{nn}^{sm}(r)}f_{R_s}^{\mathrm{NL}}(r), 
\end{align*}

\begin{enumerate}
    \item   The user is associated with the NLoS SBS and the interference is from LOS SBS.
            \begin{align}\label{NLSBSLSBS}
             p_{nl}^{ss}(r)&= \mathbb{P}[P_s^{tr}B_{s}\textcolor{blue}{G_s}h_sA_{NL}r^{-\alpha_{NL}}\geq P_s^{tr}B_{s}\textcolor{blue}{G_s}h_sA_{L}r_s^{-\alpha_L}]
             =e^{-\lambda_s\pi\left[ \left(\frac{A_\mathrm{NL}}{A_\mathrm{L}}\right)^{\frac{-1}{\alpha_{L}}} r^{\frac{\alpha_{NL}}{\alpha_{L}}}\right]^{2}}.
            \end{align}
   \textcolor{blue}{\item   The user is associated with the NLOS SBS and the interference is from LOS MBS.
   	\begin{align}\label{NLSBSLMBS}
   		p_{nl}^{sm}(r)&=\mathbb{P}[P_s^{tr}B_{s}{G_s}h_sA_{NL}r^{-\alpha_{NL}}\geq P_m^{tr}B_{m}{G_m}h_mA_Lr_m^{-\alpha_L}]
   		=e^{-\lambda_m\pi\left[\left(\frac{P_s^{tr} B_{s}{G_s}h_s A_\mathrm{NL}}{P_m^{tr} B_{m}{G_m}h_m A_{\mathrm{L}}}\right)^{\frac{-1}{\alpha_\mathrm{L}}} r^\frac{\alpha_\mathrm{NL}}{\alpha_\mathrm{L}}\right]^{2}}.
   	\end{align}
   	\item   The user is associated with the NLOS SBS and the interference is from NLOS MBS.
   	\begin{align}\label{NLSBSNLMBS}
   		p_{nn}^{sm}(r)&=\mathbb{P}[P_s^{tr}B_{s}{G_s}h_sA_{NL}r^{-\alpha_{NL}}\geq P_m^{tr}B_{m}{G_m}h_mA_{NL}r_m^{-\alpha_{NL}}]
   		=e^{-\lambda_m\pi\left[\left(\frac{P_s^{tr}B_{s}{G_s} h_s}{P_m^{tr} B_{m}{G_m} h_m}\right)^{\frac{-1}{\alpha_{\mathrm{NL}}}} r\right]^{2}}.
   \end{align}}
    
\end{enumerate}

\textcolor{blue}{Then,  the probability that the user obtain the desired LoS signal from the MBS is $F_{m}^{L}(r)=p_{1n}^{mm}(r)$ $ p_{ll}^{ms}(r) p_{ln}^{sm}(r) f_{R_m}^{\mathrm{L}}(r)$ where
	\begin{enumerate}
		\item   The user is associated with the LoS MBS and the interference is from NLoS MBS.
		\begin{align}\label{LMBSNLMBS}
			p_{ln}^{mm}(r)=\mathbb{P}[P_m^{tr}B_{m}{G_m}h_mA_Lr^{-\alpha_L}\geq P_m^{tr}B_{m}{G_m}h_mA_{NL}r_m^{-\alpha_{NL}}] =e^{-\lambda_m\pi\left[\left(\frac{A_\mathrm{L}}{A_\mathrm{NL}}\right)^{\frac{-1}{\alpha_{NL}}} r^{\frac{\alpha_{L}}{\alpha_{NL}}}\right]^{2}}
		\end{align}
		\item   The user is associated with the LoS MBS and the interference is from LoS SBS.
		\begin{align}\label{LSBSLMBS}
			p_{ll}^{ms}(r) =\mathbb{P}[P_m^{tr}B_{m}{G_m}h_mA_Lr^{-\alpha_L}\geq P_s^{tr}B_{s}{G_s}h_sA_{L}r_s^{-\alpha_{L}}]=e^{-\lambda_s\pi\left[\left(\frac{P_m^{tr}B_{m} h_m}{P_s^{tr} B_{s}h_s}\right)^{\frac{-1}{\alpha_{L}}} r\right]^{2}}
		\end{align}
		\item   The user is associated with the LoS MBS and the interference is from NLoS SBS.
		\begin{align}\label{LMBSNLSBS}
			p_{ln}^{sm}(r) =\mathbb{P}[P_m^{tr}B_{m}{G_m}h_mA_Lr^{-\alpha_L}\geq P_s^{tr}B_{s}{G_s}h_sA_{NL}r_s^{-\alpha_{NL}}]=e^{-\lambda_s\pi\left[\left(\frac{P_m^{tr}B_{m}{G_m} h_m A_\mathrm{L}}{P_s^{tr} B_{s}{G_s}h_s A_\mathrm{NL}}\right)^{\frac{-1}{\alpha_{\mathrm{NL}}}} r^{\frac{\alpha_{\mathrm{L}}}{\alpha_{\mathrm{NL}}}}\right]^{2}}
		\end{align}
\end{enumerate}
Following the same logic, the probability that the user obtain the desired NLoS signal from the MBS is $F_{m}^{NL}(r)=p_{nl}^{mm}(r)$ $ p_{nl}^{ms}(r) p_{nn}^{ms}(r) f_{R_m}^{\mathrm{NL}}(r)$}

The probability that SBS obtains the desired LOS or NLOS signal from   MBS is 
\begin{align*}
	\small F_{bh}^{L}(r)=p_{ln}^{bh}(r) f_{R_{bh}}^{\mathrm{L}}(r), 
\end{align*}
 where the probability that the SBS is associated with the LOS MBS and the interference is from NLOS MBS is
 \begin{align*}
 p_{ln}^{bh}(r)=\mathbb{P}[P_m^{tr}B_{m}\textcolor{blue}{G_m}h_mA_Lr^{-\alpha_L}\geq P_m^{tr}B_{m}\textcolor{blue}{G_m}h_mA_{NL}r_{bh}^{-\alpha_{NL}}] =e^{-\lambda_m\pi\left[\left(\frac{A_\mathrm{L}}{A_\mathrm{NL}}\right)^{\frac{-1}{\alpha_{NL}}} r^{\frac{\alpha_{L}}{\alpha_{NL}}}\right]^{2}}.
 \end{align*}
Similarly, $F_{bh}^{NL}(r)=p_{nl}^{bh}(r) f_{R_{bh}}^{\mathrm{NL}}(r)$.

\subsection{Proof of Proposition \ref{ProUserSINRCoverageProbability}}\label{Appen3:ProUserSINRCoverageProbability}

Then we first focus on the SINR distribution of a user covered by SBS:
\begin{align}
  P_{s}^{cov}(\gamma)=P_{s,L}^{cov}(\gamma)+P_{s,NL}^{cov}(\gamma),
  \end{align}
where the SINR distribution of  a user   covered by LOS/NLOS  SBS:
\begin{align}
  &P_{s,k}^{cov}(\gamma)=\mathbb{E}_{r }\left[\mathbb{P}\left[\mathrm{SINR}_{s}^{\mathrm{k}} (r ) > \gamma\right]\right]=\int_{0}^{\infty} \mathbb{P}\left[\operatorname{SINR}_{s}^{\mathrm{k}}(r) >\gamma\right] F_{s}^{\mathrm{k}}(r)\mathrm{d}r,
\end{align}
where $k=\{L\ or\ NL\}$ denotes the LOS or NLOS transmission link,

 $\gamma$ is the threshold for successful demodulation and decoding at the receiver. $\mathbb{P}\left[\mathrm{SINR}_{s}^{\mathrm{L}} (r ) > \gamma\right]$ means the probability of the event that the SINR of  the user covered by SBS is over $\gamma$  via the LOS path at distance $r$:
\begin{align}
&\mathbb{P}\left[\operatorname{SINR}_{s}^{\mathrm{k}}(r)  > \gamma\right]=\mathbb{P}\left[\frac{P_{s}^{tr}B_{s}\textcolor{blue}{G_s}h_{s} A_{\mathrm{k}} r^{-\alpha_{\mathrm{k}}}}{I_s \textcolor{blue}{+I_m}+N_{0}} > \gamma\right]\nonumber\\
&=\mathbb{P}\left[h_{s} > \frac{\gamma\left(I_s \textcolor{blue}{+I_m}+N_{0}\right)}{P_{s}^{tr} B_{s}\textcolor{blue}{G_s}A_{\mathrm{k}} r^{-\alpha_{\mathrm{k}}} }\right]\stackrel{(a)}{=} \exp \left(\frac{-\gamma N_{0}}{P_{s}^{tr} B_{s}\textcolor{blue}{G_s}A_{\mathrm{k}} r^{-\alpha_{\mathrm{k}}} }\right)  \mathcal{L}_{I_{s\textcolor{blue}{,m}}}^{\mathrm{k}}\left(\gamma r^{\alpha_{\mathrm{k}}} \right),
\end{align}

where (a) follows from small fading $h$$\sim$$\exp(1)$. Here the Rayleigh fading is considered.    $\mathcal{L}_{I_{s,m}}$  is the Laplace transform of the cumulative interference from  the SBS tier \textcolor{blue}{and the MBS tier}.
\begin{small}
\begin{align}\label{LoSInterferenceFromtheSBS}
   &\mathcal{L}_{I_{s,m}}^{\mathrm{L}}\left( \gamma r^{\alpha_{\mathrm{L}}} \right)\nonumber\\
   &\stackrel{(b)}{=}
    \textcolor{blue}{\prod \limits_{G_i}}\exp \left(-2 \pi  \lambda_{s} \textcolor{blue}{p_{G_i}}\int_{r}^{\infty}  \left(\frac{\mathcal{P}_L(u)u}{1+\frac{P_s^{tr}B_{s}\textcolor{blue}{G_{s}} A_Lr^{-\alpha_{\mathrm{L}}}}{\gamma P_s^{tr}B_{s}\textcolor{blue}{G_{i}}A_Lu^{-\alpha_{\mathrm{L}}}}} d u \right)\right)
    \textcolor{blue}{\prod \limits_{G_i}}\exp \left(-2 \pi  \lambda_{s} \textcolor{blue}{p_{G_i}}\int_{\left(\frac{A^{\mathrm{L}}}{A^{\mathrm{NL}}}\right)^{\frac{-1}{\alpha^{\mathrm{NL}}}} r^{\frac{\alpha_{\mathrm{L}}}{\alpha_{\mathrm{NL}}}}}^{\infty} 
      \left(\frac{\mathcal{P}_{\mathrm{NL}}(u)u}{1+\frac{P_s^{tr}B_{s}\textcolor{blue}{G_{s}}A_{\mathrm{L}}r^{-\alpha_{\mathrm{L}}}}{\gamma  P_s^{tr}B_{s}\textcolor{blue}{G_{i}} A_{\mathrm{NL}}u^{-\alpha_{\mathrm{NL}}}}} d u\right)\right),\nonumber\\
   &\times \textcolor{blue}{\prod \limits_{G_l}}\exp \left(-2 \pi  \lambda_{m}\textcolor{blue}{p_{G_l}}\int_{\left(d_1\right)^{\frac{-1}{\alpha_{L}}} r}^{\infty} \left(\frac{\mathcal{P}_{L}(u)u}{1+\frac{P_{s}^{tr}B_{s}\textcolor{blue}{G_{s}}A_L r^{-\alpha_{\mathrm{L}}}}{\gamma P_{m}^{tr} B_{m}\textcolor{blue}{G_{l}} A_L u^{-\alpha_{\mathrm{L}}}}} d u\right)\right)  \textcolor{blue}{\prod \limits_{G_l}}\exp \left( -2 \pi \lambda_{m}\textcolor{blue}{p_{G_l}}\int_{\left(d_2\right)^{\frac{-1}{\alpha_{\mathrm{NL}}}} r^\frac{\alpha_{\mathrm{L}}}{\alpha_{\mathrm{NL}}}}^{\infty}  \left(\frac{\mathcal{P}_{\mathrm{NL}}(u)u}{1+\frac{P_{s}^{tr} B_{s}\textcolor{blue}{G_{s}}A_\mathrm{L} r^{-\alpha_{\mathrm{L}}}}{\gamma P_{m}^{tr}B_{m}\textcolor{blue}{G_{l}}A_\mathrm{NL} u^{-\alpha_{\mathrm{NL}}}}} d u\right)\right)\nonumber,
\end{align}
\end{small}
where \textcolor{blue}{$p_{G_i}$ and $p_{G_l}$ is the probability of the antenna gain taking correspoding value from SBS interference tier and MBS interference tier.} Step (b) is based on the PGFL of PPP \cite{ReferUserClosestBSProbabilty}. $d_1=\frac{P_s^{tr}B_{s} }{P_m^{tr}B_{m}}$ and $d_2=\frac{P_s^{tr} B_{s}A_\mathrm{L}}{P_m^{tr}B_{m}A_\mathrm{NL}}$.
Following the same logic, other Laplace transforms of cumulative interference 
$  \mathcal{L}_{I_{s}}^{\mathrm{NL}}\left(\gamma r^{\alpha_{\mathrm{NL}}} \right),$
$  \mathcal{L}_{I_{bh}}^{\mathrm{L}}\left(\gamma r^{\alpha_{\mathrm{L}}} \right),$
$  \mathcal{L}_{I_{bh}}^{\mathrm{NL}}\left(\gamma r^{\alpha_{\mathrm{NL}}} \right)$
can also be  obtained.

Similarly, the SINR distribution of SBS is covered by MBS:
\begin{align}
  P_{bh}^{cov}(\gamma)
  &=P_{bh,L}^{cov}(\gamma)+P_{bh,NL}^{cov}(\gamma)\\
  &=\mathbb{E}_{r }\left[\mathbb{P}\left[\mathrm{SINR}_{bh}^{\mathrm{L}} (r ) > \gamma\right]\right]+\mathbb{E}_{r }\left[\mathbb{P}\left[\mathrm{SINR}_{bh}^{\mathrm{NL}} (r ) > \gamma\right]\right]\nonumber\\
  &=\int_{0}^{\infty} \mathbb{P}\left[\operatorname{SINR}_{bh}^{\mathrm{L}}(r)>\gamma\right] F_{bh}^{\mathrm{L}}(r)\mathrm{d} r+\int_{0}^{\infty}\mathbb{P}\left[\operatorname{SINR}_{bh}^{\mathrm{NL}}(r)>\gamma\right] F_{bh}^{\mathrm{NL}}(r) \mathrm{d} r\nonumber,
\end{align}
where
$
\mathbb{P}\left[\operatorname{SINR}_{bh}^{\mathrm{L}}(r)  > \gamma\right] = \exp \left(\frac{-\gamma N_{0}}{P_{m}^{tr} B_{m} A_{\mathrm{L}} r^{-\alpha_{\mathrm{L}}} }\right)  \mathcal{L}_{I_{bh}}^{\mathrm{L}}\left(\gamma r^{\alpha_{\mathrm{L}}} \right)
$
and
$
\mathbb{P}\left[\operatorname{SINR}_{bh}^{\mathrm{NL}}(r)  > \gamma\right]\\=  \exp \left(\frac{-\gamma N_{0}}{P_{m}^{tr} B_{m} A_{\mathrm{NL}} r^{-\alpha_{\mathrm{NL}}} }\right)  \mathcal{L}_{I_{bh}}^{\mathrm{NL}}\left(\gamma r^{\alpha_{\mathrm{NL}}} \right).
$

\subsection{\textcolor{blue}{Proof of Proposition \ref{noise_limited}}}

\textcolor{blue}{	In noise-limited scenario, the interference is close to zero, so the SINR distribution can be reduced to
	\begin{equation*}
		P_{k}^{cov}(\gamma)= \mathbb{P}[\frac{P_k g B_k G_k A r^{-\alpha}}{\sigma^{2}}>\gamma]=\mathbb{P}[\frac{r^{\alpha}}{Ag}<\frac{P_k B_k G_k}{\sigma^{2}\gamma}]
	\end{equation*}
	where $A \in \{A_\mathrm{L},A_\mathrm{NL}\},\alpha \in \{\alpha_\mathrm{L},\alpha_\mathrm{NL}\}$. Considering the effect of blockage, the intensity function of process $\Lambda=\{\frac{r^{\alpha}}{A}\}=\{\phi\}$ is calculated as
	\begin{align}
		\Lambda(\phi) = \int_{0}^{(A_\mathrm{L}\phi)^{\frac{1}{\alpha_\mathrm{L}}}}2 \pi \lambda_k u e^{-\beta u}du+\int_{0}^{(A_\mathrm{NL}\phi)^{\frac{1}{\alpha_\mathrm{NL}}}}2 \pi \lambda_k u (1-e^{-\beta u})du
	\end{align} 
	where $k \in \{m,s\}$ denotes the MBS tier or SBS tier.
	Then, the density function  is 
	\begin{align}
			\lambda(\phi) = \frac{d\Lambda(\phi)}{d\phi}
			=\frac{2\pi \lambda_k A_\mathrm{NL}}{\alpha_\mathrm{NL}}(A_\mathrm{NL}\phi)^{\frac{2}{\alpha_\mathrm{NL}}-1} 
			+\frac{2\pi \lambda_k(A_\mathrm{L}\phi)^{\frac{2}{\alpha_\mathrm{L}}-1}}{\alpha_\mathrm{L}e^{\beta (A_\mathrm{L}\phi)^{\frac{1}{\alpha_\mathrm{L}}}}} 
			-\frac{2\pi \lambda_k(A_\mathrm{NL}\phi)^{\frac{2}{\alpha_\mathrm{NL}}-1}}{\alpha_\mathrm{NL}e^{\beta (A_\mathrm{NL}\phi)^{\frac{1}{\alpha_\mathrm{NL}}}}}
	\end{align}
	Then, the joint distribution fuction and probability density function of $\{\phi,g\}$ can be given as 
	\begin{align}
		\mathbb{P}[\frac{\phi}{g}\le\xi ]&=\mathbb{P}[g \ge\frac{\phi}{\xi}] = 1-F_g(\frac{\phi}{\xi}) \\
		\rho(\phi,\xi)&=\frac{d(1-F_g(\frac{\phi}{\xi}))}{d \xi}=\frac{\phi}{\xi^2}f_g(\frac{\phi}{\xi})\overset{(a)}{=}\frac{\phi}{\xi^2}e^{-\frac{\phi}{\xi}}
	\end{align}
	where step (a) is based on the Rayleigh fading channel with exponential distribution ($h\sim\exp(1)$).
	Based on the displacement theorem, the density function of process $\Xi = \{\frac{\phi}{g}\}=\{\xi\}$ can be obtained as 
	\begin{align}
		\lambda_{\Xi}(\xi)&=\int_{0}^{\infty}\lambda(\phi)\rho(\phi,\xi)d\phi \nonumber \\
		&=\int_{0}^{\infty}\left( \frac{2\pi \lambda_k A_\mathrm{NL}}{\alpha_\mathrm{NL}}(A_\mathrm{NL}\phi)^{\frac{2}{\alpha_\mathrm{NL}}-1} 
		+\frac{2\pi \lambda_k(A_\mathrm{L}\phi)^{\frac{2}{\alpha_\mathrm{L}}-1}}{\alpha_\mathrm{L}e^{\beta (A_\mathrm{L}\phi)^{\frac{1}{\alpha_\mathrm{L}}}}} 
		-\frac{2\pi \lambda_k(A_\mathrm{NL}\phi)^{\frac{2}{\alpha_\mathrm{NL}}-1}}{\alpha_\mathrm{NL}e^{\beta (A_\mathrm{NL}\phi)^{\frac{1}{\alpha_\mathrm{NL}}}}}\right)\frac{\phi}{\xi^2}e^{-\frac{\phi}{\xi}}d\phi \nonumber\\
		&= \frac{2\pi\lambda_kA_\mathrm{NL}^{\frac{2}{\alpha_{\mathrm{NL}}}}}{\alpha_\mathrm{NL}}\xi^{\frac{2}{\alpha_\mathrm{NL}}-1}\Gamma(\frac{2}{\alpha_\mathrm{NL}}+1)+\frac{2\pi \lambda_k}{\xi^2}(H_\mathrm{L}-H_\mathrm{NL})
	\end{align}
	where $H_i=\frac{A_i^{\frac{2}{\alpha_i}-1}}{\alpha_i}\int_{0}^{\infty}\phi^{\frac{2}{\alpha_i}}\exp(-\beta(A_i\phi)^{\frac{1}{\alpha_i}}-\frac{\phi}{\xi})d\phi,i\in\{\mathrm{L},\mathrm{NL}\}$.
	Now, based on the complementary void function of PPP, the CDF of $\xi$ is given by
	\begin{align}
		F_{\Xi}(\xi_0)=\mathbb{P}[\xi<\xi_0]&=1-\mathbb{P}[\Xi(\xi_0)=0]=1-\exp(-\int_{0}^{\xi_0}\lambda_{\Xi}(\xi)d\xi)\nonumber \\
		&=1-\exp(-\pi \lambda_k A_\mathrm{NL}^{\frac{2}{\alpha_{\mathrm{NL}}}}\Gamma(\frac{1}{\alpha_\mathrm{NL}}+1)\xi_0^{\frac{2}{\alpha_\mathrm{NL}}}-2 \pi \lambda_kY(\xi_0))
	\end{align} 
	where $Y(\xi_0)=\int_{0}^{\xi_0}\frac{A_\mathrm{L}^{\frac{2}{\alpha_\mathrm{L}}-1}}{ \alpha_\mathrm{L} \xi^2}\int_{0}^{\infty}\phi^{\frac{2}{\alpha_\mathrm{L}}}\exp(-\beta\phi^{\frac{1}{\alpha_\mathrm{L}}}-\frac{\phi}{\xi}) d\phi d\xi-\int_{0}^{\xi_0}\frac{A_\mathrm{NL}^{\frac{2}{\alpha_\mathrm{NL}}-1}}{ \alpha_\mathrm{NL} \xi^2}\int_{0}^{\infty}\phi^{\frac{2}{\alpha_\mathrm{NL}}}\exp(-\beta\phi^{\frac{1}{\alpha_\mathrm{NL}}}-\frac{\phi}{\xi}) d\phi d\xi
	$. Then we can obtain that
	\begin{align}
		P_{k}^{cov}(\gamma)&=\mathbb{P}[\frac{r^{\alpha}}{g}<\frac{P_k B_kG_k}{\sigma^{2}\gamma}]=F_{\Xi}(\frac{P_kB_k G_k}{\sigma^{2}\gamma})\nonumber \\
		&=1-\exp\left(-\pi \lambda_k A_\mathrm{NL}^{\frac{2}{\alpha_{\mathrm{NL}}}}\Gamma\left(\frac{1}{\alpha_\mathrm{NL}}+1\right)\left(\frac{P_kB_k G_k}{\sigma^{2}\gamma}\right)^{\frac{2}{\alpha_\mathrm{NL}}}-2 \pi \lambda_kY\left(\frac{P_kB_k G_k}{\sigma^{2}\gamma}\right)\right)
	\end{align}
}


\end{document}